\documentclass[aps,pra,longbibliography,superscriptaddress,twocolumn]{revtex4-1}

\usepackage[T1]{fontenc}
\usepackage[utf8]{inputenc}

\usepackage{float}
\usepackage{graphicx}
\usepackage{epsfig,epstopdf}
\usepackage{comment}
\usepackage{thm-restate}
\usepackage{mathtools}

\usepackage{bm,bbm}
\newcommand{\LParen}{ \bm{(} }
\newcommand{\RParen}{ \bm{)} }
\newcommand{\compU}{ \mathcal{C}_{\mathrm{u}} }
\newcommand{\compState}{ \mathcal{C}_{\mathrm{state}} }

\newcommand{\rr}{\mathbb{R}}
\newcommand{\id}{\mathbb{1}}

\usepackage[english]{babel}

\PassOptionsToPackage{hyphens}{url}
\usepackage[hyphenbreaks]{breakurl}
\usepackage{times}

\usepackage{amsmath}
\usepackage{bbm}
\usepackage{amsfonts}
\usepackage{amssymb}
\usepackage{amsthm}
\usepackage{mathbbol}
\usepackage{mathtools}

\usepackage{float}
\usepackage{graphicx}
\usepackage{epsfig,epstopdf}
\usepackage{xcolor}
\usepackage{tikz}
\usepackage{pgffor}

\usepackage{times}

\usepackage{enumerate}
\usepackage{enumitem}

\usepackage{pgfplots}
\usetikzlibrary{pgfplots.groupplots}
\pgfplotsset{compat=1.11}
\usepgfplotslibrary{fillbetween}

\usepackage{tikz}
\usetikzlibrary{
  shapes,
  shapes.geometric,
	trees,
	matrix,
  positioning,
    pgfplots.groupplots,
  }

\PassOptionsToPackage{pdftex,pdfpagelabels}{hyperref}
\usepackage{hyperref}
\hypersetup{
    colorlinks=true, linktocpage=true, pdfstartpage=3, pdfstartview=FitV,
    breaklinks=true, pdfpagemode=UseNone, pageanchor=true, pdfpagemode=UseOutlines,
    plainpages=false, bookmarksnumbered, bookmarksopen=true, bookmarksopenlevel=1,
    hypertexnames=true, pdfhighlight=/O,
    urlcolor=blue, linkcolor=blue, citecolor=black,
}   


\newtheorem{theorem}{Theorem}
\newtheorem{proposition}{Proposition}
\newtheorem{definition}{Definition}
\newtheorem{lemma}{Lemma}

\newcommand{\ket}[1]{\vert{#1}\rangle}

\definecolor{comment}{rgb}{0,.4,1}

\definecolor{jens}{rgb}{0,.8,.5}

\definecolor{roberto}{rgb}{0.4,.0,0.6}

\definecolor{nicole}{rgb}{.7,.1,0}

\definecolor{dominik}{rgb}{0.4,.0,0.6}

\definecolor{dg}{rgb}{0,0,0}
\newcommand{\new}[1]{{\color{dg}#1}}

\definecolor{naga}{rgb}{0.5,.3,0.7}
\newcommand{\nt}[1]{{\color{naga}}}

\usepackage{phfcc}

\phfMakeCommentingCommand[initials={PhF},color=phfcc1]{phf}

\begin{document}
\title{Linear growth of quantum circuit complexity}
\author{Jonas~Haferkamp}
\affiliation{Dahlem Center for Complex Quantum Systems, Freie Universit{\"a}t Berlin, 14195 Berlin, Germany}
\affiliation{Helmholtz-Zentrum Berlin f{\"u}r Materialien und Energie, 14109 Berlin, Germany}
\author{Philippe~Faist}
\affiliation{Dahlem Center for Complex Quantum Systems, Freie Universit{\"a}t Berlin, 14195 Berlin, Germany}
\author{\ Naga\ B.~T.~Kothakonda}
\affiliation{Dahlem Center for Complex Quantum Systems, Freie Universit{\"a}t Berlin, 14195 Berlin, Germany}
\affiliation{Institute for Theoretical Physics, University of Cologne, 50937 Cologne, Germany}
\author{Jens~Eisert}
\affiliation{Dahlem Center for Complex Quantum Systems, Freie Universit{\"a}t Berlin, 14195 Berlin, Germany}
\affiliation{Helmholtz-Zentrum Berlin f{\"u}r Materialien und Energie, 14109 Berlin, Germany}
\author{Nicole~Yunger~Halpern}
\affiliation{ITAMP, Harvard-Smithsonian Center for Astrophysics, Cambridge, MA 02138, USA}
\affiliation{Department of Physics, Harvard University, Cambridge, MA 02138, USA}
\affiliation{Research Laboratory of Electronics, Massachusetts Institute of Technology, Cambridge, Massachusetts 02139, USA}
\affiliation{Center for Theoretical Physics, Massachusetts Institute of Technology, Cambridge, Massachusetts 02139, USA}
\affiliation{Institute for Physical Science and Technology, University of Maryland, College Park, MD 20742, USA}

\maketitle

{\bf 
	Quantifying quantum states' complexity is a key problem in various subfields of science, from quantum computing to black-hole physics. We prove a prominent conjecture by Brown and Susskind about how random quantum circuits' complexity increases. Consider constructing a unitary from Haar-random two-qubit quantum gates. Implementing the unitary exactly requires a circuit of some minimal number of gates---the unitary's exact circuit complexity. We prove that this complexity grows linearly with the number of random gates, with unit probability, until saturating after exponentially many random gates. Our proof is surprisingly short, given the established difficulty of lower-bounding the exact circuit complexity. Our strategy combines differential topology and elementary algebraic geometry with an inductive construction of Clifford circuits.}\smallskip
    
%Quantum
\new{Complexity} is a pervasive concept at the intersection of \new{computer science,} quantum computing, quantum many-body systems, and black-hole physics.
In general, complexity quantifies the resources required to implement a computation.
For example, a Boolean function's complexity can be defined as the minimal number of gates, chosen from a given gate set, necessary to evaluate the function.
In quantum computing, the circuit model provides a natural measure of complexity \new{for pure states and unitaries}:
A unitary transformation's \new{quantum circuit} complexity is the size, \new{measured with the number of gates,}
% \new{counted as the number of gates,} 
% NYH: Mild typo -- sizes aren't counted
of the smallest circuit that effects the unitary.
Similarly, a pure state's \new{quantum circuit}  
complexity definable is the size of the smallest circuit that produces the
state from a product state.

\new{Quantum circuit complexity, by quantifying the minimal
size of any circuit that implements a given unitary, is closely related to computational notions of complexity. 
The latter quantify the difficulty of
solving a given computational task with a quantum computer and
determine quantum complexity classes.
Yet quantum circuit complexity can
subtly differ from computational notions of quantum complexity:
The computational notion depends on
the difficulty of finding the circuit.
In the following, we refer to quantum circuit complexity
as ``quantum complexity'' for convenience.}

Quantum complexity has risen to prominence recently, due to connections between gate complexity and holography in high-energy physics, in the context of the anti-de-Sitter-space/conformal-field-theory (AdS/CFT) correspondence \cite{susskind2016computational,stanford2014complexity,brown2016complexity,PhysRevD.97.086015,bouland2019computational}.
%
%I think that Referee 1 means the following statements with generic:
%Quantum complexity has applications in numerous subfields of quantum physics, including the foundations of quantum computing. Quantum computational complexity is rooted in circuit complexity: A problem is deemed ``easy'' if 
%it is soluble with a circuit whose size
%grows polynomially with the input's size. 
%The problem is deemed ``hard'' if the circuit's size %scales exponentially. 
%Quantum complexity also features in the definition of phases of matter.
%For instance, a topological phase} %\new{requires a complexity that scales at %least linearly in the system size}
%is defined in terms of a high-complexity quantum state. %
%Recently, close connections have been discovered between gate complexity and holography in high-energy physics, in the context of the anti-de-Sitter-space/conformal-field-theory (AdS/CFT) correspondence \cite{susskind2016computational,stanford2014complexity,brown2016complexity,PhysRevD.97.086015,bouland2019computational}.
\new{In the bulk theory, a wormhole's volume grows steadily for exponentially long times. In contrast, in boundary quantum theories, local observables tend to thermalize much more quickly. 
This contrast is known as the \emph{wormhole-growth paradox}~\cite{susskind2016computational}. It appears to contradict the AdS/CFT correspondence, which postulates
a mapping of physical operators
%a dictionary
between the bulk theory and a quantum boundary theory. 
% As a resolution, it is proposed that the quantity corresponding to the wormhole's volume is not a local observable but the quantum circuit complexity~\cite{susskind2016computational}.
A resolution has been proposed in the ``complexity equals volume'' conjecture:
The wormhole's volume is conjectured to be dual not to a local quantum observable, but to the boundary state's quantum complexity~\cite{stanford2014complexity}.}
% The correspondence's boundary state has a complexity proportional to the volume behind the event horizon of a black hole in the bulk geometry.}
Similarly, the ``complexity equals action'' conjecture
posits that a holographic state's complexity is dual to a certain  space-time region's action \cite{PhysRevLett.116.191301}.
%These deep physical connections have motivated studies of quantum complexity as a means of illuminating quantum many-body systems' complex behaviors.

\new{
A counting argument reveals that the vast majority of unitaries
have near-maximal complexities~\cite{PhysRevLett.106.170501,knill1995approximation}. Yet
% The arguably most straightforward notion of circuit complexity
% focuses on exact implementations of a unitary~\cite{nielsen2005geometric,dowling2008geometry,nielsen2006quantum}: the number of quantum gates required to implement a given unitary. Defined so, quantum complexity cannot be computed efficiently. 
%Relatedly,
lower-bounding the quantum complexity is a long-standing open problem
in quantum information theory. The core difficulty is that the gates performed early in a circuit may partially cancel with gates
performed later. One can rarely rule out the existence of a ``shortcut'',
a seemingly unrelated but smaller circuit that %accidentally
generates the same unitary.
Consequently, quantum-gate--synthesis algorithms, which decompose a given
unitary into gates, run for times exponential in the system
size \cite{TComment}.
%Several attempts have been made to lower-bound unitaries' circuit
Approaches to lower-bounding unitaries' quantum
complexities include Nielsen's geometric
picture~\cite{nielsen2005geometric,nielsen2006quantum,Nielsen_06_Optimal,dowling2008geometry,Entanglement}.
%and unitary $t$-designs~\cite{brandao2019models}.
}

% A few claims have been \new{proven} about circuit complexity.
% For instance, polynomially deep circuits implement only a 
% \new{small fraction of the unitaries
% \cite{PhysRevLett.106.170501}.} 
% Most unitaries have complexities exponentially large in the system size, according to a counting argument~\cite{knill1995approximation}.
A key question in the study of quantum complexity is the following.
Consider constructing deeper and deeper circuits for an $n$-qubit system, by applying random two-qubit gates. At what rate does the circuit complexity increase?
Brown and Susskind conjectured that quantum circuits' complexity generically grows linearly for an exponentially long time~\cite{PhysRevD.97.086015,susskind2018black}. 
\new{Intuitively, the conjecture is that most circuits are
%'s minimal description consists is
fundamentally ``incompressible'': 
No substantially shorter quantum circuit effects the same unitary.
% Hence most circuits are close to the simplest descriptions of the unitaries that they implement.
}
\new{Quantum complexity, if it grows linearly with a generic circuit's depth, strongly supports the ``complexity equals volume'' conjecture as a proposal to the wormhole-growth paradox~\cite{susskind2016computational,stanford2014complexity}.}
% \nicole{Is the resolution Susskind's alone, or do collaborators deserve credit, too?} of the wormhole-growth paradox~\cite{susskind2016computational}. Also, is ``complexity equals volume'' not proposed in~\cite{stanford2014complexity}, rather than in~\cite{susskind2016computational}?}
% \jo{~\cite{susskind2016computational} was uploaded to the arXiv a few months earlier. Adam Bill and Umesh also cite ~\cite{susskind2016computational} in their pseudorandomness paper. Of course they might be mistaken. They also call it Susskind's Volume=Complexity conjecture.}
\new{The conjecture therefore implies that complexity growth is as generic as thermalization \cite{1408.5148, ngupta_Silva_Vengalattore_2011} and operator growth \cite{SwingleScrambling, Shenker} (the spreading of an initially local operator's support in the Heisenberg picture). 
However, in contrast to easily measurable physical quantities, which
thermalize rapidly, complexity grows for an exponentially long time.}
Brown and Susskind have supported their conjecture using Nielsen's
\new{geometric approach (Figure~\ref{figure:complexity-growth}\textbf{b})}~\cite{nielsen2005geometric,nielsen2006quantum,Nielsen_06_Optimal,dowling2008geometry}.
Further evidence for the conjecture \new{has} arisen from counting arguments~\cite{roberts2017chaos}.
\new{Brandao \emph{et al.}~\cite{brandao2019models} recently proved a key result about quantum complexity's growth under random circuits.
The authors leveraged the mathematical toolbox of \emph{$t$-designs}, finite collections of unitaries that approximate completely random unitaries~\footnote{A $t$-design is a probability distribution, over unitaries, whose first $t$ moments equal the Haar measure's moments~\cite{gross_evenly_2007,dankert_exact_2009,brandao_local_2016}. \new{The Haar measure is the unique unitarily invariant probability measure over a compact group.}}.
Ref.~\cite{brandao2019models} proved that quantum complexity robustly grows polynomially in a random circuit's size.
% a robust polynomial growth of quantum complexity in the size of a random circuit.  In the regime of high local dimension,
The complexity's growth was shown to be linear in
the circuit's size if the local Hilbert-space dimension is large.
}
% Growth-of-complexity studies % \new{have} 
% advanced further in \new{Ref.}~\cite{brandao2019models},
% which connected a strong notion of quantum complexity to 
% unitary $t$-designs
% (finite collections of unitaries that approximate completely random unitaries
% \footnote{A $t$-design is a probability distribution, over unitaries, whose first $t$ moments equal the Haar measure's moments~\cite{gross_evenly_2007,dankert_exact_2009}. \new{The Haar measure is the unique unitarily invariant probability measure over a compact group.}}).

We prove that \new{a random circuit's complexity}
grows linearly with time (with the number of gates applied).
We consider unitaries
% The unitaries under consideration <-- NYH: By whom?
constructed from \new{quantum circuits composed of} Haar-random two-qubit gates. %\new{The gates are arranged in a generalization, detailed below, of common architectures.} %%%PhF--I feel this information is too vague to be of any use at this point.
% one of several discussed meaningful architectures}.
%The proof is surprisingly short, given the difficulty of lower-bounding complexity:
%\new{An arbitrary circuit's complexity cannot be calculated efficiently.}
%PhF---I don't think it's good to repeat quasi word-for-word a sentence we already stated in the abstract. Instead, let's offer a sentence with an intuition on our revolutionary proof technique:
\new{The focus of our proof is the set of unitaries that can be generated with a fixed arrangement of gates. We show that this set's dimension, which we call \emph{accessible dimension}, serves as a good proxy for the quantum complexity of almost every unitary in the set.}
% \new{Any attempt to compute a given circuit's complexity for a given circuit must be doomed to fail, as it is not a problem that awaits an efficient solution.}
Our bound on the complexity holds for all random circuits described above, with probability $1$.
Instead of invoking unitary designs~\cite{brandao2019models} or Nielsen's geometric approach~\cite{nielsen2005geometric,nielsen2006quantum,Nielsen_06_Optimal,dowling2008geometry},
we employ elementary aspects of differential topology and algebraic geometry, combined with an inductive construction of Clifford circuits.
Clifford circuits play a pivotal role in quantum computing, as circuits that can easily be implemented fault-tolerantly~\cite{Gottesman_99_Heisenberg,Gottesman_99_Fault}.

\begin{figure}
 	\centering
 	\includegraphics{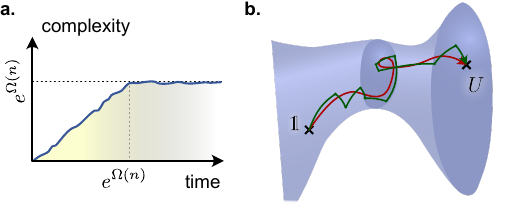}
 	\caption{\new{\textbf{a.}~The complexity has been conjectured to grow linearly under random quantum circuits until times exponential in the number $n$ of qubits~\cite{PhysRevD.97.086015}.
 	%The \textit{complexity ramp}: Evolution of quantum complexity, as
 	%conjectured in Ref.~\cite{PhysRevD.97.086015}.
 	\textbf{b.}~The blue region depicts part of the space of $n$-qubit
 	unitaries. A unitary $U$ has a complexity that we define as
 	the minimal number of two-qubit gates necessary to effect $U$ (green jagged path; each path segment represents a gate).
 	Nielsen's complexity~\cite{nielsen2005geometric,nielsen2006quantum,Nielsen_06_Optimal,dowling2008geometry},
 	involved in Ref.~\cite{PhysRevD.97.086015},
 	attributes a high metric cost to directions 
 	associated with nonlocal operators.
 	In this geometry, the unitary's complexity is the shortest path that connects $\mathbb{1}$ to $U$ (red line).  
 	Nielsen's geometry suggests the toolbox of differential
 	geometry, avoiding circuits' discreteness.
 	The circuit complexity upper-bounds Nielsen's
 	complexity; opposite bounds hold for approximate circuit complexity~\cite{dowling2008geometry}.}}
 	\label{figure:complexity-growth}
 \end{figure}
 
This work is organized as follows. 
First, we introduce the setup and definitions.
Second, we present the main result,
the complexity's exponentially long linear growth.
We present a high-level overview of the proof third.
The key mathematical steps follow,
in the methods section.
Two corollaries follow: an extension to random arrangements of gates and an extension to slightly imperfect gates.
In the discussion, we compare our results with known results
and explain our work's implications for various subfields of quantum physics.
Finally, we discuss the opportunities engendered by this work.
In Appendix~A of \new{Ref.}~\cite{suppmaterial}, we review elementary algebraic geometry required for the proof. 
Proof details appear in Appendix~B.
We elaborate on states' complexities in Appendix~
C.
We prove two corollaries in Appendices~D and~E.
Finally, we compare notions of circuit complexity in Appendix~F.

\emph{Preliminaries.} This work concerns a system of $n$ qubits. 
For convenience, we assume that $n$ is even. 
We simplify tensor-product notation as 
$\ket{0^k} :=  \ket{0}^{\otimes k}$, for $k = 1, 2, \ldots, n$;
and $\mathbb{1}_k$ denotes the $k$-qubit identity operator.
Let $U_{j,k}$ denote a unitary gate that operates on qubits $j$ and $k$. 
Such gates need not couple the qubits together and need not be geometrically local.
An \emph{architecture} is an arrangement of
some fixed number $R$ of gates [Figure~\ref{figure:brickwork}(a)].

\begin{definition}[Architecture]\label{def_architecture}
An architecture is a directed acyclic graph that contains $R \in \mathbb{Z}_{>0}$ vertices (gates). 
Two edges (qubits) enter each vertex, and two edges exit.
\end{definition}
Figures~\ref{figure:brickwork}(b) and~\ref{figure:brickwork}(c)
illustrate example architectures governed by our results.

\begin{figure*}
  \includegraphics{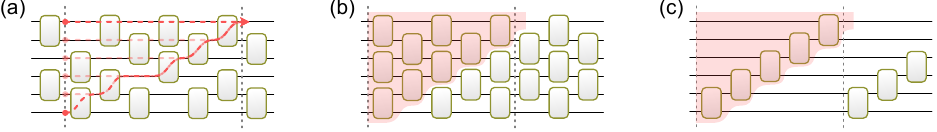}
  \caption{
    Our result relies on architectures and their backwards light cones.
    \textbf{(a)}~An \emph{architecture} specifies how $R$ 2-qubit gates are arranged in an $n$-qubit circuit.
    The gates need not be applied to neighboring qubits, though they are depicted this way for convenience.
    \new{Our result involves blocks with the following property:}
    The block contains a qubit reachable from each other qubit via a path (red dashed line), possibly unique to the latter qubit, that
    passes only through gates in the \new{block}.
    \textbf{(b)}~The \emph{brickwork architecture} interlaces layers of gates on a one-dimensional (1D) chain. 
    In a 1D architecture with geometrically local gates, such as the brickwork architecture, each \new{block} has a \new{backwards} light cone (light-red region) that touches the qubit chain's edges. In the brickwork architecture, a minimal \new{backwards-light-cone--containing block} consists of
    ${\sim} n^2$ gates.
    \textbf{(c)}~The staircase architecture, too, acts on a 1D qubit chain. The circuit consists of layers in which $n-1$ gates
    act on consecutive qubit pairs.  A minimal \new{backwards-light-cone--containing block consists of} $n-1$ gates.
}
  \label{figure:brickwork}
\end{figure*}

	\begin{itemize}
\item	A \emph{brickwork} is the architecture of any circuit formed as follows:
	Apply a string of two-qubit gates: 
	$U_{1,2} \otimes U_{3,4}\otimes \ldots \otimes U_{n-1,n}$.
	Then, apply a staggered string of gates, 
	as shown in Fig.~\ref{figure:brickwork}(b).
	Perform this pair of steps $T$ times total, using possibly different gates each time.

	\item 
	A \emph{staircase} is the architecture of any circuit formed as in Fig.~\ref{figure:brickwork}(c):
	Apply a stepwise string of two-qubit gates: 
	$U_{n,n-1}U_{n-2,n-1} \ldots U_{2,1}$. 
	Repeat this process $T$ times, using possibly different gates each time. 
	\end{itemize}
\noindent
The total number of gates in the brickwork architecture, as in the staircase architecture, is $R=(n-1)T$.
Our results extend to more-general architectures,
e.g., the architecture depicted in Fig.~\ref{figure:brickwork}(a)
and architectures of non-nearest-neighbor gates.
Circuits of a given architecture can be formed randomly.

\begin{definition}[Random quantum circuit] \label{def_randomquantumcircuits}
   Let $A$ denote an arbitrary architecture. 
   A probability distribution can be induced over
   the architecture-$A$ circuits as follows:
   For each vertex in $A$, draw a gate Haar-randomly from $\mathrm{SU}(4)$. Then, contract the unitaries along 
   the edges of $A$.
   Each circuit so constructed is called a \emph{random quantum circuit}.
\end{definition}
\noindent
Implementing a unitary with the optimal gates,
in the optimal architecture, concretizes the notion of complexity.
\begin{definition}
[Exact circuit complexities]
\label{def_Circuit_Comp}
Let $U\in\mathrm{SU}(2^n)$ denote an $n$-qubit unitary.
The (exact) circuit complexity $\compU(U)$ is the least number of 
two-qubit gates in any circuit that implements $U$.
Similarly, let $\ket{\psi}$ denote a pure quantum state vector.
The (exact) state complexity $\compState(|\psi\rangle)$ is the least number $r$ of two-qubit gates 
$U_1, U_2, \dots ,U_r$, arranged in any architecture, such that 
$U_1 U_2  \dots U_r|0^n\rangle=|\psi\rangle$.
\end{definition}

\new{We now define a \emph{backwards light cone}, 
a concept that helps us focus on sufficiently connected circuits.}
% the circuits we consider are sufficiently well connected.}
\new{Consider creating two vertical cuts in a circuit 
(dashed lines in Fig.~\ref{figure:brickwork}).
The gates between the cuts form a \emph{block}.
We say that a block contains a \emph{backwards light cone} if 
some qubit $t$ links to each other qubit $t'$
via a directed path of gates (a path that may be unique to $t'$).
The backwards light cone consists of the gates in the paths.}
%A slice need not be a time slice: 
%The gates in a slice need not all be simultaneously implementable.

\emph{Main result: Linear growth of complexity in random quantum circuits.} Our main result is a lower bound on the complexities of random unitaries and states. The bound holds with unit probability. 

\begin{restatable}[Linear growth of complexity]{theorem}{lineargrowth} \label{theorem:linear growth}
Let $U$ denote a unitary implemented by
a random quantum circuit in an architecture formed \new{by concatenating $T$ %$ = R/L$
blocks of $\new{\leq} L$ gates each,
each block containing a backwards light cone.}
The unitary's circuit complexity is lower-bounded as
 \begin{equation}
    \label{eq_Linr_Growth}
    \compU(U)\geq \frac{R}{9L} -\frac{n}{3}\ ,
 \end{equation}
 with unit probability, until the number of gates grows to
 $T \geq 4^n-1$. 
 The same bound holds for $\compState(U|0^n\rangle)$, until 
 $T \geq 2^{n+1}-1$.
\end{restatable}

The theorem governs all architectures \new{that contain enough backwards light cones}.
The brickwork architecture forms a familiar special case.
Let us choose for a brickwork's \new{blocks} to contain 
$2n$ of the columns in Fig.~\ref{figure:brickwork}(b).
Each \new{block} contains $L = n(n-1)$ gates
(in the absence of periodic boundary conditions), 
yielding the lower bound
$\compU(U)\geq \frac{R}{9n(n-1)}-\frac{n}{3} \, .$
Another familiar example is the staircase architecture.
A staircase's \new{blocks} can have the least $L$ possible, $n-1$, 
which yields the strongest bound. 

\emph{High-level overview of the proof of Theorem~\ref{theorem:linear growth}.}  
\begin{figure}
	\centering
	\includegraphics{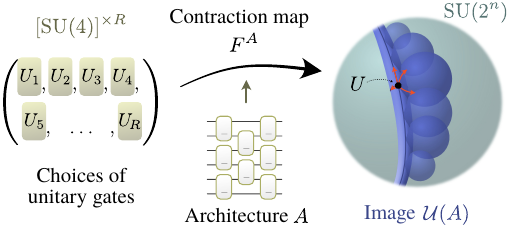}
	\caption{The $R$-gate architecture $A$ is associated with a contraction map $F^A$. $F^A$ maps a list of input gates (a point in $[\mathrm{SU}(4)]^{\times R}$) to an $n$-qubit unitary $U$ in $\mathrm{SU}(2^n)$. The unitary results from substituting the gates into the architecture. $F^A$ has an image $\mathcal{U}(A)$, which consists of the unitaries implementable with the architecture.
		$A$ has an \emph{accessible dimension}, $d_A$, equal to the dimension of $\mathcal{U}(A)$. Our core technical result is that $d_A$ grows linearly with $R$. To bridge this result to complexity, consider an arbitrary architecture $A'$ formed from fewer gates than a constant fraction of $R$. Such an architecture's accessible dimension satisfies $d_{A'} < d_A$, as we show. Therefore, every unitary in $\mathcal{U}(A)$ has a complexity linear in $R$, except for a measure-0 set.
		The proof relies on algebraic geometry. A key concept is the rank of $F^A$ at a point. The rank counts the local degrees of freedom in the image (orange arrows).}
	\label{fig_Contraction_map}
\end{figure}
Consider fixing an $R$-gate architecture $A$,
then choosing the gates in the architecture.
The resulting circuit implements some $n$-qubit unitary.
All the unitaries implementable with $A$ form a set $\mathcal{U}(A)$.
Our proof relies on properties of 
$\mathcal{U}(A)$---namely, on
the number of degrees of freedom in $\mathcal{U}(A)$.
We define this number as the architecture's \emph{accessible dimension}, $d_{A} = \dim \LParen \mathcal{U}(A) \RParen$
(Fig.~\ref{fig_Contraction_map}).
The following section contains a formal definition;
here, we provide intuition.
As the $n$-qubit unitaries form a space of dimension $4^n$, 
$d_A \in [0, 4^n]$.
The greater the $d_A$, the more space $\mathcal{U}(A)$ fills in the set of $n$-qubit unitaries.
\new{%
Considering $\mathcal{U}(A)$ circumvents the intractability of calculating a unitary's circuit complexity.
%Therefore, any proof of Theorem~\ref{theorem:linear growth} must be based on the randomness of the ensemble of random quantum circuits.
To better understand the form of $\mathcal{U}(A)$, 
we study the set's dimension, which is the accessible dimension.  Importantly, the accessible dimension enables us to compare the sets $\mathcal{U}(A)$ generated by different architectures.
% a differing accessible dimension necessary implies that 
Distinct accessible dimensions imply that
the lower-dimensional set has measure zero in the higher-dimensional set.
% We must show that the ensemble of deep random quantum circuits has too large
% a support to consist of many low-complexity unitaries.
% For continuous ensembles, the accessible dimension turns out to 
% quantify this randomness. 
As a proxy for quantum complexity,
the accessible dimension plays a role similar to $t$-designs
in Refs.~\cite{brandao2019models,brandao_efficient_2016}}.
Our first technical result lower-bounds sufficiently connected architecture's accessible dimension:
\begin{proposition}[Lower bound on accessible dimension]
  \label{thm:lower-bound-on-dA}
  Let $A_T$ denote
  \new{an architecture formed by concatenating $T$
  blocks of $\new{\leq} L$ gates each,
  each block containing a backwards light cone.}
  %an architecture with $R=TL$ gates.  
  %Assume that $A_T$ \new{results from concatenating blocks, each
  %containing a backwards light cone and consisting of $\leq L$ gates}.
  The architecture's accessible dimension is lower-bounded as
  \begin{align}
    d_{A_T} \geq T \geq \frac{R}{L} \ .
    \label{eq:dA-lower-bound}
  \end{align}
\end{proposition}

We can upper-bound $d_{A}$, for an arbitrary architecture $A$, by counting parameters.
To synopsize the argument in Appendix~B:
Fifteen real parameters specify each 2-qubit unitary. 
Each qubit shared by two unitaries makes 3 parameters redundant. Hence
\begin{align}
  d_A \leq 9R + 3n\ .
  \label{eq:dA-upper-bound}
\end{align}
The accessible dimension reaches its maximal value, $4^n$, after a number of gates exponential in $n$. 
Similarly, the circuit complexity reaches its maximal value after exponentially many gates.
This parallel suggests $d_A$ as a proxy for the circuit complexity.
The next section rigorously justifies the use of $d_A$ as a proxy.

The proof of Theorem~\ref{theorem:linear growth} revolves around the accessible dimension $d_{A_T}$ of a certain $R$-gate architecture $A_T$.  
The main idea is as follows. 
Let $R'$ be less than a linear fraction of $R$. 
More specifically, let $9R'+3n < T = R/L$.
For every $R'$-gate architecture $A'$, 
$d_{A'} < d_{A_T}$ holds by a combination of~\eqref{eq:dA-lower-bound}
and~\eqref{eq:dA-upper-bound}.
Consequently, Appendix~B in \new{Ref.}~\cite{suppmaterial} shows, $\mathcal{U}(A')$ has zero probability in $\mathcal{U}(A_T)$, according to the measure in Definition~\ref{def_randomquantumcircuits}.
Therefore, almost every unitary $U\in \mathcal{U}(A_T)$ has a complexity greater than the greatest possible $R'$.
Inequality~\eqref{eq_Linr_Growth} follows.

\emph{Discussion.}
We have proven a prominent physics conjecture proposed by Brown and Susskind for random quantum circuits~\cite{PhysRevD.97.086015,susskind2018black}:
A local random circuit's quantum complexity grows linearly in the number of gates until reaching a value exponential in the system size.
To prove this conjecture, we introduced a novel
technique for bounding complexity.
The proof rests on our connecting the quantum complexity to the accessible dimension, the dimension of the set of unitaries implementable with a given architecture (arrangement of gates).
Our core technical contribution is
a lower bound on the accessible dimension.
The bound rests on techniques from differential topology
and algebraic geometry.

To the best of our knowledge, Theorem~\ref{theorem:linear growth} is the first rigorous demonstration of the linear growth of random qubit circuits' complexities for exponentially long times. The bound holds until the complexity reaches $\compU(U)=\Omega(4^n)$---the scaling, up to polynomial factors, of the greatest complexity achievable by any $n$-qubit unitary~\cite{NielsenChuang}.
A hurdle has stymied attempts to prove that local random circuits' quantum complexity grows linearly:
Most physical properties (described with, e.g., local observables or correlation functions) reach fixed values in times subexponential in the system size.
One must progress beyond such properties
to prove that the complexity grows linearly at superpolynomial times.
We overcome this hurdle by identifying the accessible dimension as a proxy for the complexity.

Theorem~\ref{theorem:linear growth} complements another rigorous insight about complexity growth. 
In Ref.~\cite{brandao2019models},
the linear growth of complexity is proven in the limit of large local dimension $q$ and for a strong notion of quantum circuit complexity, with help from Ref.~\cite{hunter2019unitary}.
Furthermore, depth-$T$ random qubit circuits have complexities that scale as $\Omega(T^{1/11})$
until $T=\exp\LParen \Omega(n) \RParen$~\cite{brandao2019models,brandao_local_2016}.
The complexity scales the same way for other types of random unitary evolutions, such as a continuous-time evolution under a stochastically fluctuating Hamiltonian~\cite{RandomHamiltonians}. 
Finally, Ref.~\cite{brandao2019models} addresses bounds on convergence to unitary designs~\cite{brandao_local_2016, hunter2019unitary,nakata2017efficient,haferkamp2020quantum,RandomHamiltonians,haferkamp2020improved}, translating these bounds into results about circuit complexity.
Theorem~\ref{theorem:linear growth} is neither stronger nor weaker than the results of Ref.~\cite{brandao2019models}, which govern a more operational notion of complexity---how easily $U |0^n\rangle\langle 0^n|U^\dagger$ can be distinguished from the maximally mixed state.

%While our results focus on circuits for qubits, we
%anticipate no significant challenges in extending our
%results to circuits for qudits of finite dimension $D>2$.
%The complexity lower bound might scale unfavorably
%in $D$, however. The reason is, 
%our complexity bound's coefficient stems from the dimension of the group $\mathrm{SU}(4)$, which would be replaced with $\mathrm{SU}(D^2)$. As the relevant group's dimension grows, the coefficient---and the lower bound---shrinks.
%However, we expect the resulting bound to scale
%linearly with the number of gates.  The main challenge, we anticipate, is to generalize our Clifford-circuit construction, outlined in Fig.~\ref{fig:PerturbationCliffordCircuits}(e). An alternative approach would be to embed each qudit in $\log_2(D)$
%qubits and consider gates that act jointly on $\log_2(D)$ qubits.
%......

\new{Our work is particularly relevant} to the holographic context surrounding the Brown-Susskind conjecture.
There, random quantum circuits are conjectured to serve as proxies for chaotic quantum dynamics
generated by local time-independent Hamiltonians~\cite{PhysRevX.8.021014}.
Ref.~\cite{HaydenBlackHoles} has introduced this conjecture into black-hole physics, and Ref.~\cite{susskind2016computational} discussed the conjecture in the context of holography.
A motivation for invoking random circuits is,
random circuits can be analyzed more easily
than time-independent--Hamiltonian dynamics.
Time-independent--Hamiltonian dynamics are believed to be mimicked also by time-fluctuating Hamiltonians~\cite{RandomHamiltonians} and 
by random ensembles of Hamiltonians.
%Random ensembles have enabled progress in the study of scrambling, or the spread of local perturbations through many-body entanglement, via the Sachdev-Ye-Kitaev model~\cite{FastScramblingConjecture,Shenker,PhysRevD.97.086015}.
%Studying random circuits' complexity
%is expected to illuminate the complexity of chaotic
%Hamiltonian dynamics similarly.
Furthermore, complexity participates in analogies
with thermodynamics, such as a second law of
quantum complexity~\cite{PhysRevD.97.086015}.  Our techniques can
be leveraged to construct an associated resource theory of
complexity~\cite{resourcepaper}.

In the context of holography, \emph{thermofield double states}' complexities have attracted
recent interest~\cite{susskind2016computational,Eternal,BigComplexity,EntanglementNotEnough}.
Thermofield double states are pure bipartite quantum states for which each subsystem's reduced state is thermal.  In the context of holography, thermofield double states are dual to eternal black holes in anti-de-Sitter space~\cite{Eternal}. 
Such a black hole's geometry consists of two sides 
connected by a wormhole, or Einstein-Rosen bridge. 
The wormhole's volume grows for a time exponential in the number of degrees of freedom of the boundary theory~\cite{susskind2016computational,PhysRevD.97.086015}.
As discussed above, random quantum circuits are expected to capture the (presumed Hamiltonian) dynamics behind the horizon.
If they do, the growth of the wormhole's volume is conjectured to match the growth of the boundary state's complexity \cite{susskind2016computational,stanford2014complexity,PhysRevD.97.086015}; both are expected to reach a value exponentially large in the number of degrees of freedom. 
Our results govern the random circuit that serves as a proxy for the dynamics behind the horizon.
That random circuit's complexity, our results show strikingly,
indeed grows to exponentially large values. 
\new{This conclusion reinforces
%provides strong
the evidence that quantum circuit complexity is the right quantity
with which to resolve the wormhole-growth
paradox~\cite{susskind2016computational}.} 

%Our proof hinges on our architectures' connectedness---on an architecture's consisting of causal slices.
%One might find an alternative proof technique to bypass this assumption. However, a notion related to causal slices appears
%in the switchback effect in holography~\cite{PhysRevD.97.086015,arXiv:1408.2823,PhysRevD.90.126007}.
%Consider producing a state by inputting $\ket{0^n}$ into a generic quantum circuit.
%Now, consider sequentially applying the gates' inverses, starting with the last gate's inverse.
%The state's complexity is expected to decrease until reaching zero.  Suppose, however, that a qubit is perturbed during this process.
%The state's complexity is no longer expected to decrease to zero.
%Instead, the complexity will start increasing again.  The delay between the perturbation and the complexity's return to increasing is the \emph{switchback effect}.  
%It is of the order of the time necessary for the perturbation to grow, in the Heisenberg picture, to reach the full system size.
%A causal slice is defined strikingly similarly: The slice enables 
%a qubit, if time is reversed,
%to affect all the qubits at the causal slice's input.
%This similarity indicates that our definition of a causal slice, used here
%as a proof technique, might have broader relevance.

\emph{Outlook.} Our main result governs exact circuit complexity.
In \new{Ref.}~\cite[Cor.~2]{suppmaterial}, we generalize the result to
\new{a slightly robust notion of} % approximate  %%PhF -- isn't "approximate" redundant, being implied by "robust notion of complexity"?
circuit complexity.
There, the complexity depends on our tolerance of the error in the implemented unitary. Yet, the error tolerance can be uncontrollably small.
The main challenge in extending our results to 
approximate complexity is,
the accessible dimension crudely characterizes the set of unitaries implementable with a given architecture. 
Consider attempting to enlarge this set to include
all the $n$-qubit unitaries that lie close to the set in some norm.  The enlarged set's dimension is $4^n$.
The reason is, the enlargement happens in all
directions of $\mathrm{SU}(2^n)$.
Therefore, our argument does not work as for the exact complexity.
Extending our results to approximations therefore offers an opportunity for 
future work. 
Approximations may also 
illuminate random circuits as instruments for identifying quantum advantages~\cite{neill_blueprint_2017,Supremacy}; 
they would show that a polynomial-size quantum circuit cannot be compressed substantially while achieving a good approximation.
These observations motivate an uplifting of the present work to robust notions of quantum circuit complexity (see, e.g., Ref.~\cite{brandao2019models}).
A possible uplifting might look as follows.
Let $A$ denote an $R$-gate architecture,
and let $A'$ denote an $R'$-gate architecture.
Suppose that the accessible dimensions obey $d_{A'} < d_A$.  
A unitary implemented with $A$ has no chance of 
occupying the set $\mathcal{U}(A')$, 
which has a smaller dimension than $\mathcal{U}(A)$.
Consider enlarging $\mathcal{U}(A')$ to include the unitaries that lie $\epsilon$-close, for some $\epsilon>0$. If $\mathcal{U}(A')$ is sufficiently
smooth and well-behaved,
we expect the enlarged set's volume, intersected with $\mathcal{U}(A)$, to scale as ${\sim}\,\epsilon^{d_A - d_{A'}}$.
Furthermore, suppose that unitaries implemented with $A$ are distributed sufficiently evenly in $\mathcal{U}(A)$ 
[rather than being concentrated close to $\mathcal{U}(A')$].
All the unitaries in $\mathcal{U}(A)$ except a small fraction
${\sim}\,\epsilon^{d_A-d_{A'}}$ could not lie in $\mathcal{U}(A')$.  
We expect, therefore, that all the unitaries in $\mathcal{U}(A)$ except a fraction ${\sim}\,\epsilon^{d_A-d_{A'}}$ have $\epsilon$-approximate complexities greater than $R'$.

\new{A related opportunity is a proof that Nielsen's geometric complexity measure grows linearly under random circuits.
Such a proof 
% Proving linear growth of Nielsen's geometric notion of complexity under random circuits 
likely requires a more refined characterization of $\mathcal{U}(A)$ than its dimension.
The quantum complexity in Theorem~\ref{theorem:linear growth}
does not lower-bound Nielsen's complexity.
% Nielsen's complexity is not lower-bounded by the quantum complexity that appears in Theorem~\ref{theorem:linear growth}, meaning that 
Hence our main results do not immediately imply a similar
bound for Nielsen's complexity.
However, proving the approximate circuit complexity's linear growth would suffice to lower-bound Nielsen's complexity, due to known inequalities between Nielsen's complexity and the circuit complexity [Fig.~\ref{figure:complexity-growth}(b)]
(e.g., Ref.~\cite{dowling2008geometry}).}

\new{We expect our machinery to be applicable to random processes
that more closely reflect a variety of systems that are studied in the
many-body physics community.
%plausible physical dynamics.
Examples include randomly fluctuating dynamics
\cite{RandomHamiltonians},
which implement random quantum circuits when Trotterized,
and thermofield-double states undergoing random ``shocks''~\cite{shenker2014black,shenker2014multiple,bouland2019computational}.
Additionally, \emph{hybrid circuits}---random unitary circuits punctuated by intermediate measurements---have recently attracted much interest~\cite{Li_19_Measurement,Skinner_19_Measurement,Chan_19_Unitary},
as the amount of entanglement present in such systems
appear to undergo phase transitions induced by the rate at which they
are measured.
%Entanglement signals measurement-induced phase transitions undergone by these circuits.
A generalization of the accessible dimension to such systems might
reveal to what extent circuit complexity, as a measure of entanglement in deep dynamics, undergoes similar phase transitions. 
}
We hope that the present work, by innovating machinery for addressing complexity, stimulates further quantitative studies of holography, scrambling, and chaotic quantum dynamics.

\emph{Acknowledgements.}
We thank Aram Harrow and Richard K\"ung for discussions and thank P\'{e}ter Varj\'{u} for introducing us to the algebraic geometrical methods used in this paper. N.~Y.~H.~thanks Shira Chapman, Michael Walter, and the other organizers of the 2020 Lorentz Center workshop ``Complexity: From quantum information to black holes'' for inspiration. This work has been funded by the DFG (EI 519/14-1, CRC 183, for which this is an inter-node Berlin-Cologne project, and FOR 2724), \new{by the Einstein Research Foundation,
the FQXi,} and by an NSF grant for the Institute for Theoretical Atomic, Molecular, and Optical Physics at Harvard University and the Smithsonian Astrophysical Observatory. Administrative support was provided by the MIT CTP.

\emph{Author contributions.} J.~H.~developed the basic proof technique. All authors wrote the manuscript and established the results.

\emph{Data and code availability statement.} No data or code has been generated in this work.

\bibliographystyle{plain}
%\bibliography{BigReferences57}

\newpage 

\onecolumngrid

\cleardoublepage
\setcounter{page}{0}
\setcounter{equation}{0}
\setcounter{footnote}{0}
\thispagestyle{empty}
\begin{center}
	\textbf{\large Methods}\\
	% \vspace{2ex}
\end{center}

Having overviewed the proof at a high level, we fill in the key mathematics.
Three points need clarifying.
First, we must rigorously define the accessible dimension, or the dimension of $\mathcal{U}(A)$, which is not a manifold.
Second, we must prove Proposition~\ref{thm:lower-bound-on-dA}. Finally, we must elucidate steps in the proof of Theorem~\ref{theorem:linear growth}.
We address these points using the toolbox of algebraic geometry.
We associate with every $R$-gate architecture $A$ a \emph{contraction map} 
$F^{A}: \mathrm{SU}(4)^{\times R} \to \mathrm{SU}(2^n)$.
This function maps a list of gates to an $n$-qubit unitary.
The unitary results from substituting the gates into the architecture $A$ (Fig.~\ref{fig_Contraction_map}).
The map contracts every edge (qubit) shared by two vertices (gates) in $A$.

The image of $F^A$ is the set $\mathcal{U}(A)$
of unitaries implementable with the architecture $A$.
$\mathcal{U}(A)$ is a \emph{semialgebraic set},
consisting of the solutions to a finite set of polynomial equations and inequalities over the real
numbers (see Appendix~A for a review).
That $\mathcal{U}(A)$ is a semialgebraic set follows from
the Tarski-Seidenberg principle, a deep result in semialgebraic geometry (Appendix~A).
A semialgebraic set's dimension quantifies the degrees of freedom needed to describe the set locally.
More precisely, a semialgebraic set decomposes into manifolds. 
The greatest dimension of any such manifold
equals the semialgebraic set's dimension.
The dimension of $\mathcal{U}(A)$ is 
the architecture $A$'s accessible dimension.
More restricted than a semialgebraic set is an algebraic set, which consists of the solutions to a finite set of polynomial equations.

Just as the contraction map's image will prove useful, so will the map's rank, defined as follows. Let 
$   x=(U_1, U_2, \ldots, U_R)
\in \mathrm{SU}(4)^{\times R}$
denote an input into $F^A$, such that the $U_j$ denote two-qubit gates.
The map's \emph{rank} at $x$ is the rank of a matrix that approximates $F^A$ linearly around $x$ (the rank of the map's Jacobian at $x$).
The rank is low at $x$ if perturbing $x$
can influence the $n$-qubit unitary only along few directions in $\mathrm{SU}(2^n)$.

Crucially, we prove that $F^A$ has the same rank 
throughout the domain, except on a measure-zero set, where $F^A$ has a lesser rank. 
The greater, ``dominating'' rank is the dimension of $\mathcal{U}(A)$.
To formalize this result, let $E_r$ denote the locus of points at which $F^A$ has a rank of $r \geq 0$.
Let $E_{<r} = \bigcup_{r'<r} E_{r'}$ denote the set of points where $F^A$ has a lesser rank. 
Let $r_{\mathrm{max}}$ denote the maximum rank achieved by $F^A$ at any point $x$. 
We prove the following lemma in Appendix~B, using the dimension theory of real algebraic sets.

\begin{restatable}[Low-rank locus]{lemma}{locus}\label{lemma:locus} 
	The low-rank locus $E_{<r_{\mathrm{max}}}$ is an algebraic set  of measure $0$ and so is closed (in the Lie-group topology). Equivalently, $E_{r_{\mathrm{max}}}$ is an open set of measure $1$.
	Consequently, $d_A = r_{\mathrm{max}}$.
\end{restatable}
\noindent
Lemma~\ref{lemma:locus} guarantees that the contraction map's rank equals the accessible dimension $d_A$ almost everywhere in $\mathcal{U}(A)$.

We now turn to the proof of Proposition~\ref{thm:lower-bound-on-dA}. 
The rank $r$ of $F^A$ at each point $x$ lower-bounds $r_{\mathrm{max}}$, by definition. 
Consider an architecture $A_T$ of $T$ \new{blocks, each containing a backwards light cone}. 
We identify an $x$ at which $r$ is lower-bounded by
a quantity that grows linearly with $R$
(the number of gates in the architecture $A_T$).
We demonstrate the point's existence
by constructing circuits from Clifford gates.

Consider a choice 
$x=(U_1, U_2, \ldots, U_R)
\new{=:} 
(U_j)_j$ of unitary gates.
Perturbing a $U_j$ amounts to appending an infinitesimal unitary:
$U_j \mapsto \tilde{U}_j = e^{i\epsilon H} U_j$.
The $H$ denotes a 2-qubit Hermitian operator, and 
$\epsilon \in \rr$.
$H$ can be written as a linear combination of 2-qubit Pauli strings $S_k$.
(An $n$-qubit Pauli string is a tensor product of $n$ single-site operators, each of which is a Pauli operator [$X$, $Y$, or $Z$] or the identity, $\id_1$.
The $4^n$ $n$-qubit Pauli strings form a basis for the space of $n$-qubit Hermitian operators.)
Consider perturbing each gate $U_j$ using a combination of all 15 nontrivial 2-qubit Pauli strings
[Fig.~\ref{fig:PerturbationCliffordCircuits}(a)]:
$x=(U_j)_j \mapsto \tilde{x} = 
\LParen \exp(i\sum_{k=1}^{15} \epsilon_{j,k} S_k) \,
U_j \RParen_j$, wherein
$\epsilon_{j,k} \in \rr$. 
The perturbation $x\mapsto \tilde{x}$ causes a perturbation $U=F^{A_T}(x)\mapsto \tilde{U}=F^{A_T}(\tilde{x})$  of the image under $F^{A_T}$.
The latter perturbation is, to first order,
$\partial_{\epsilon_{j,k}} \tilde{U}\bigr|_{\epsilon_{j,k}=0}$. 
This derivative can be expressed as the original circuit
with the Pauli string $S_k$ inserted immediately after the gate $U_j$ [Fig.~\ref{fig:PerturbationCliffordCircuits}(b)].

\begin{figure*}
	\centering
	\includegraphics{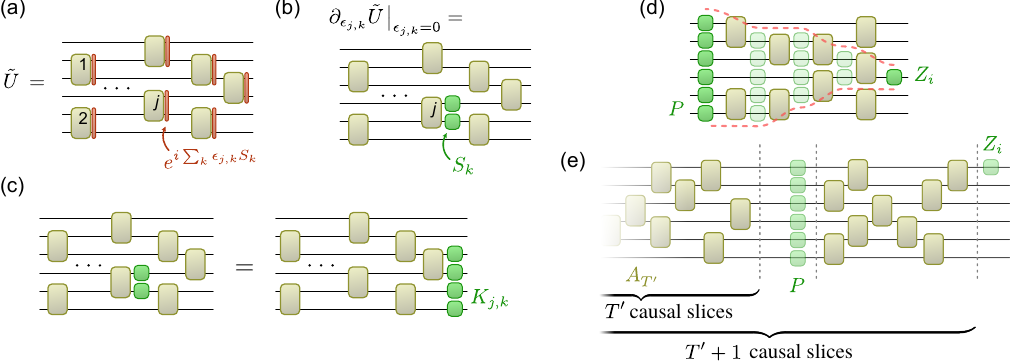}
	\caption{Our core technical result is a lower bound on the accessible dimension (see Fig.~\ref{fig_Contraction_map}).
		We prove this bound using a construction based on 
		Clifford circuits.  
		\textbf{(a)}~Each gate $U_j$ is perturbed with a unitary 
		$e^{i \epsilon_{j,k} S_k}$,
		generated by a 2-qubit Pauli operator $S_k$ 
		and parameterized with an infinitesimal 
		$\epsilon_{j,k} \in \rr$.
		Perturbing the gate perturbs the $n$-qubit unitary,
		turning $U$ into $\tilde{U} \approx U$.
		\textbf{(b)}~A key quantity is the derivative of $\tilde{U}$ with respect to a parameter $\epsilon_{j,k}$, evaluated at $U$.
		Taking this derivative is equivalent to inserting the Pauli string $S_k$ immediately after the gate $U_j$. 
		\textbf{(c)}~The derivative depicted in panel (b) is equivalent to following the circuit with a Hermitian operator $K_{j,k}$
		[Eq.~\eqref{eq_Decomp_Perturb}].  The operator $K_{j,k}$ results from
		conjugating $S_k$ with the gates after $U_j$.
		If the circuit consists of Clifford gates, then $K_{j,k}$ is a Pauli string, since Clifford gates map the Pauli strings to Pauli strings.  
		Therefore, a perturbation of $U_j$ in the direction of $S_k$ results in a perturbation of the resulting unitary $U$ in the direction of $K_{j,k}$ in $\mathrm{SU}(2^n)$.
		\textbf{(d)}~The following is true of every \new{backwards-light-cone--containing block} and every Pauli string $P$ (leftmost green squares): The \new{block}'s gates can be chosen to be Cliffords that map $P$ to a single-site $Z$.
		The Clifford gates first map $P$ to a Pauli string that acts nontrivially on fewer qubits (pale green squares), then to a Pauli string on fewer qubits, and so on until the Pauli string dwindles to one qubit (rightmost green square).
		\textbf{(e)}~Our lower bound is proven by recursion.
		Consider an architecture $A_{T'}$, formed from $T' < 4^n - 1$ \new{backwards-light-cone--containing blocks},
		whose accessible dimension is $\geq T'$.  
		There exist gates $U_1, U_2, \ldots,U_{R'}$ such that
		that $T'$ linearly independent Pauli operators 
		$K'_{j_m,k_m}$ (wherein $m = 1, 2, \ldots, T'$)
		result from perturbing the gates, as described in (a)--(c).
		Consider a Pauli operator $P$ that is not in $\{ K'_{j_m,k_m} \}$.
		We can append to $A_{T'}$ a \new{backwards-light-cone--containing block}, formed from Clifford gates,
		that maps $P$ to a single-site $Z$, as depicted in panel (d). 
		This $Z$ is an important direction in $\mathrm{SU}(2^n)$:
		Consider perturbing the \new{block}'s final gate
		via the procedure in (a)--(c).
		The image $\mathcal{U}(A_{T'})$ is perturbed, as a result, in the direction $Z$.
		Thus, $T'+1$ linearly independent Pauli operators (the operators $K'_{j_m,k_m}$ and $P$) result from 
		perturbing gates in the extended circuit.
		Therefore, the extended circuit's accessible dimension is $\geq T'+1$.
	}
	\label{fig:PerturbationCliffordCircuits}
\end{figure*}

The rank of $F^{A_T}$ at $x$ is the number of parameters $\epsilon_{j,k}$ needed to parameterize a general perturbation of $U=F^{A_T}(x)$ within the image set $\mathcal{U}(A_T)$.
To lower-bound the rank of $F^{A_T}$ at a point $x$, we need only show
that $\geq r$ parameters $\epsilon_{j,k}$ perturb $F^{A_T}(x)$ in independent directions.
To do so, we express the derivative as
\begin{align}
\label{eq_Decomp_Perturb}
\partial_{\epsilon_{j,k}} F^{A_T}(\tilde{x}) \bigr|_{\epsilon_{j,k}=0}
= K_{j,k} F^{A_T}(x) \ ,
\end{align}
wherein $K_{j,k}$ denotes a Hermitian operator
[Fig.~\ref{fig:PerturbationCliffordCircuits}(c)].
$K_{j,k}$ results from conjugating $S_k$,
the Pauli string inserted into the circuit after gate $U_j$,
with the later gates. 
The physical significance of $K_{j,k}$ follows from perturbing the gate $U_j$ in the direction $S_k$
by an infinitesimal amount $\epsilon_{j,k}$.
The image $F^{A_T}(x)$ is consequently perturbed,
in $\mathrm{SU}(2^n)$, in the direction $K_{j,k}$.

We choose for the gates $U_j$ to be Clifford operators. 
(The Clifford operators are the operators that map the Pauli strings to the Pauli strings, to within a phase, via conjugation.
For every Clifford operator $C$ and Pauli operator $P$,
$CPC^\dag$ equals a phase times a Pauli string~\cite{calderbank1997quantum,calderbank1998quantum,gottesman1997stabilizer,bolt1961cliffordI,bolt1961cliffordII}.)
As a result, the operators $K_{j,k}$ are Pauli strings (up to a phase).
Two Pauli strings are linearly independent if and only if they differ.
For Clifford circuits, therefore, we can easily verify whether perturbations
of $x$  
cause independent perturbation directions in
$\mathrm{SU}(2^n)$: We need only show that the resulting operators $K_{j,k}$ are distinct.

We apply that fact to prove Proposition~\ref{thm:lower-bound-on-dA}, 
using the following observation.
Consider any Pauli string $P$ and any \new{backwards-light-cone--containing block} of any architecture. 
We can insert Clifford gates into the \new{block} such that
two operations are equivalent:
(i) operating on the input qubits with $P$ before the extended \new{block} and (ii) operating with the extended \new{block}, then with a one-qubit $Z$. Figure~\ref{fig:PerturbationCliffordCircuits}(d) depicts the equivalence, which follows from the structure of \new{backwards light cones}. 
We can iteratively construct a Clifford unitary that reduces the Pauli string's weight 
until producing a single-qubit operator.
See Appendix~\ref{app_Prove_Lemma1} for details.

We now prove Proposition~\ref{thm:lower-bound-on-dA} by recursion.  Consider an $R'$-gate architecture $A_{T'}$ formed from 
$T' < 4^n - 1$ \new{blocks, each containing a backwards light cone} and each of $\leq L$ gates.  
% Let $x'$ denote a list of Clifford gates slotted into $A_{T'}$.
% Assume that $F^{A_{T'}}$ has a rank $\geq T'$ at $x'$. 
Assume that there exists a list $x'$ of Clifford gates,
which can be slotted into $A_{T'}$, such that
$F^{A_{T'}}$ has a rank $\geq T'$ at $x'$.
Consider appending a \new{backwards-light-cone--containing block} to $A_{T'}$.
The resulting architecture corresponds to a contraction map
whose rank is $\geq T'+1$, we show.

By assumption, we can perturb $x'$ such that its image, $F^{A_{T'}}(x')$, is perturbed in $\geq T'$ independent directions in $\mathrm{SU}(2^n)$.
These directions can be represented by Pauli operators 
$K'_{j_m,k_m}$, wherein $m=1, 2, \ldots, T'$,
by Eq.~\eqref{eq_Decomp_Perturb}.
Let $P$ denote any Pauli operator absent from $\{ K'_{j_m,k_m} \}$.  
We can append to $A_{T'}$ a \new{backwards-light-cone--containing block}, forming an architecture $A_{T'+1}$ of 
$T'+1$ \new{backwards light cones}.
We design the new \new{block} from Clifford gates such that
two operations are equivalent:
(i)~applying $P$ to the input qubits before the extended \new{blocks} 
and (ii)~applying the extended \new{block}, 
then a single-site $Z$. 
We denote by $x''$ the list of gates in $x'$ augmented with the gates in
the extended \new{block}.
Conjugating the $K'_{j_m,k_m}$ with the new \new{block} yields operators $K''_{j_m,k_m}$, for $m=1, 2, \ldots, T'$.
They represent the directions in which the image $F^{A_{T'+1}}(x'')$
is perturbed by the original perturbations of $A_{T'}$.
The $K''_{j_m,k_m}$ are still linearly independent Pauli operators.
Also, the $K''_{j_m,k_m}$ and the single-site $Z$ form an
independent set,  
because $P$ is not in $\{ K'_{j_m,k_m} \}$.
Meanwhile, the single-site $Z$ is a direction in which
the last \new{block}'s final gate can be perturbed.
The operators $K_{j_m,k_m}$, augmented with the single-site $Z$, therefore span $T'+1$ independent directions along which $F^{A_{T'+1}}(x'')$ can be perturbed.
Therefore, $T'+1$ lower-bounds the rank of $F^{A_{T'+1}}$.

We apply the above argument recursively, starting from an
architecture that contains no gates.
The following result emerges:
Consider any architecture $A_T$ that consists of $T$ \new{backwards-light-cone--containing blocks}.
At some point $x$, the map $F^{A_T}$ has a rank lower-bounded by $T$.
Lemma~\ref{lemma:locus} ensures that the
same bound applies to $d_{A_T}$.

To conclude the proof of Theorem~\ref{theorem:linear growth}, 
we address an architecture $A'$ whose accessible dimension satisfies $d_{A'} < d_{A_T}$. 
Consider sampling a random circuit with the architecture $A_T$.
We must show that the circuit has a zero probability of implementing a unitary in $\mathcal{U}(A')$.
To prove this claim, we 
invoke the constant-rank theorem:
Consider any map whose rank is constant locally---in any open neighborhood of any point in the domain. In that neighborhood, the map is equivalent to a projector, up to a diffeomorphism.
We can apply the constant-rank theorem to the contraction map: 
$F^{A_T}$ has a constant rank throughout 
$E_{r_{\mathrm{max}}}$, by Lemma~\ref{lemma:locus}.
Therefore, $F^{A_T}$ acts locally as a projector throughout $E_{r_{\mathrm{max}}}$---and
so throughout $\mathrm{SU}(4)^{\times R}$,
except on a measure-0 region, by Lemma~\ref{lemma:locus}.
Consider mapping an image back, through a projector, to a preimage.
Suppose that the image forms a subset of dimension lower than the whole range's dimension. 
The backward-mapping just adds degrees of freedom to the image.
Therefore, the preimage locally has a dimension less than the domain's dimension.
Hence the preimage is of measure 0 in the domain.
We use the unitary group's compactness to elevate this local statement to the global statement in Theorem~\ref{theorem:linear growth}.
\newpage

\cleardoublepage
\setcounter{page}{0}
\setcounter{equation}{0}
\setcounter{footnote}{0}
\thispagestyle{empty}
\begin{center}
	\textbf{\large Supplementary Material for ``Linear growth of quantum circuit complexity''}\\
	\vspace{2ex}
	J.\ Haferkamp, P.\ Faist, N.\ B.\ T.\ Kothakonda, J.\ Eisert and N.\ Yunger Halpern
	% \vspace{2ex}
\end{center}

\appendix
\section{Algebraic and semialgebraic sets}
\label{app_Alg_Review}

For convenience, we review elementary aspects of algebraic geometry over the real numbers. We apply these properties in the proof of Theorem~\ref{theorem:linear growth}. 
Ref.~\cite{bochnak2013real} contains a more comprehensive
treatment.

\begin{definition}[Algebraic set]
\label{def_Alg_Subset}
	A subset $V\subseteq \rr^m$ is called an \emph{algebraic set}, or an \emph{algebraic variety}, if, 
	for a set of polynomials $\{f_j\}_j$,
	\begin{equation}
	V=\{x\in \rr^m| f_j(x)=0 \} .
	\end{equation}
	A subset $V'\subseteq V$ is called an \emph{algebraic subset} if $V'$ is an algebraic set. 
	We call a subset $W\subseteq \rr^m$ a \emph{semialgebraic set} if, for sets $\{f_j\}_j$ and $\{g_k\}_k$ of polynomials,
	\begin{equation}
	W=\{x\in \rr^m| f_j(x)= 0, g_k(x)\leq 0 \} .
	\end{equation}
\end{definition}
\noindent
A natural topology on algebraic sets is the Zariski topology.

\begin{definition}[Zariski topology]
	Let $V$ denote an algebraic set.
	The Zariski topology is the unique topology whose closed sets
	are the algebraic subsets of $V$.
\end{definition}
A traditional definition of ``dimension'' for algebraic sets involves irreducible sets.
\begin{definition}[Irreducible set]
Let $X$ denote a topological space. 
$X$ is called \emph{irreducible} if it is not the union of two proper closed subsets. 
\end{definition}
\begin{definition}[Dimension of algebraic sets]\label{def:dimension}
Let $V$ be an algebraic set that is irreducible with respect to the Zariski topology.
The dimension of $V$ is the maximal length $d$ of any chain $V_0\subset V_1\subset \dots \subset V_d$ of distinct nonempty irreducible algebraic subsets of $V$.
\end{definition}
The relevant algebraic sets in the proof of Theorem~\ref{theorem:linear growth} are $\mathrm{SU}(4)^{\times R}$ and $\mathrm{SU}(2^n)$.
Our interest in semialgebraic 
sets stems from the following principle.
In the following, we refer to a function $F: \mathbb{R}^n \to \mathbb{R}^m$
as a \emph{polynomial map} if its entries are polynomials in the entries of its
input.

\begin{theorem}[Tarski-Seidenberg principle]\label{theorem:Tarski}
	Let $F:\rr^n\to \rr^m$ be a polynomial map. 
	If $W$ is a semialgebraic set, so is $F(W)$.
\end{theorem}
\noindent
The Tarski-Seidenberg principle applies to the map that contracts sets of quantum gates. 
This application is important for us, because it provides a natural notion of dimension for the contraction map's image.

All semialgebraic sets (and hence all algebraic sets) decompose into smooth manifolds.

\begin{theorem}[Stratification of semialgebraic sets]
\label{thm_Stratification}
	If $W$ is a semialgebraic set, then
	$W=\bigcup_{j=1}^N M_j$, 
	wherein each $M_j$ denotes a smooth manifold.
	If $W$ is an algebraic set of dimension $d$ in the sense of
	Definition~\ref{def:dimension},
	then 
	$\max_j \{ \mathrm{dim} (M_j) \} = d$.
\end{theorem}
\noindent
\new{This $\max_j\{\dim(M_j)\}$ does not depend on the decomposition chosen.
This independence motivates the following definition:

\begin{definition}[Dimension of semialgebraic sets]
Let $W$ denote a semialgebraic set, such that $W = \bigcup_{j=1}^N M_j$, wherein each $M_j$ denotes a manifold. The greatest dimension of any manifold, $\max_j \{ {\rm dim} (M_j) \}$, is the semialgebraic set's dimension.
\end{definition}
\noindent
This definition generalizes Definition~\ref{def:dimension}, 
due to Theorem~\ref{thm_Stratification}.
One more fact about semialgebraic sets' dimensions will prove useful:}

\begin{lemma}[Dimension of an image]\label{lemma:dimensionimage}
		Let $F:\rr^n\to \rr^m$ be a polynomial map. 
	If $W$ is a dimension-$d$ semialgebraic set,
	$F(W)$ is of dimension $\leq d$.
\end{lemma}
\noindent
The bound follows from combining the results of
Ref.~\cite[Prop.~2.8.7]{bochnak2013real} with the results of
Ref.~\cite[Prop.~2.8.6]{bochnak2013real}. (Ref.~\cite{bochnak2013real} invokes a semialgebraic mapping, which encompasses polynomial maps.)

\section{Proof of the main theorem and lemmata} 
\label{app_Prove_Lemma1}

In this appendix, we prove Lemma~\ref{lemma:locus}, Lemma~\ref{lemma:pointofhighrank}, and the main theorem.
The proofs rely on the topics reviewed in Appendix~\ref{app_Alg_Review},
as well as the following notation and concepts.
In differential geometry, the rank of $F^A$ at the point 
$x = (U_1, U_2, \ldots, U_R)$ is defined as 
the rank of the derivative $D_xF^A$.
Mapping lists of gates to unitaries, $F$ is a complicated object.
We can more easily characterize a map from real numbers to real numbers. 
Related is a map from Hermitian operators to 
Hermitian operators:
An $n$-qubit state evolves under a Hamiltonian represented by 
a $2^n \times 2^n$ Hermitian operator, which has 
$(2^n)^2 = 4^n$ real parameters.
Therefore, for convenience, we shift focus from unitaries to their Hermitian generators.
We construct a map whose domain is the algebra 
$\mathfrak{su}(4)^{\times R} 
\simeq  \rr^{15R}$
that generates SU$(4)^{\times R}$.
The range is the set of $n$-qubit Hermitian operators, 
$\mathfrak{su}(2^n) \simeq \rr^{4^n}$.
We construct such a map from three steps,
depicted by the dashed lines in Fig.~\ref{fig_Chart}.

\begin{figure}[hbt]
\centering
\includegraphics[width=.45\textwidth, clip=true]{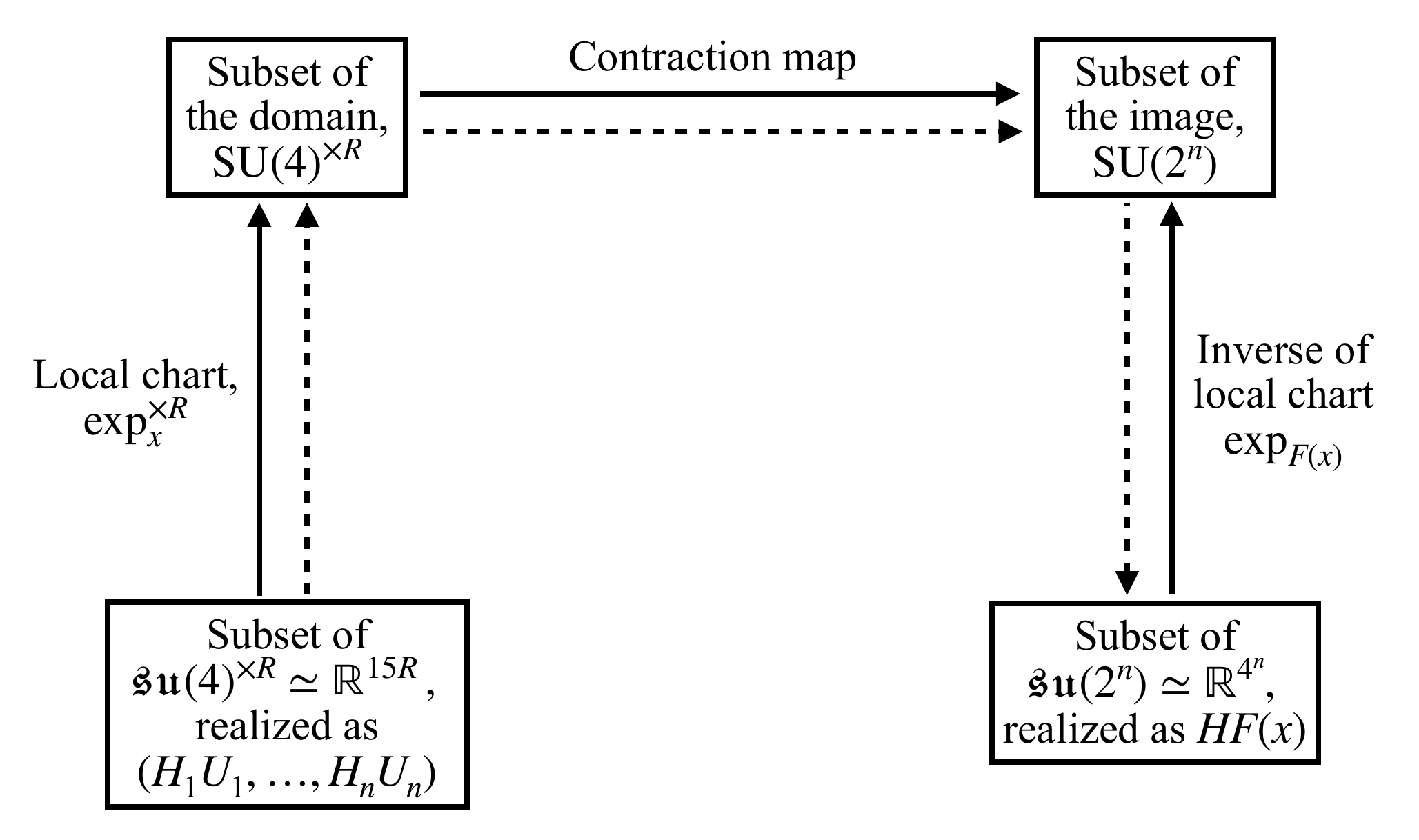}
\caption{Three-part map used in the proof of Lemma~\ref{lemma:locus}.
$H_j$ denotes the $j^{\rm th}$ two-qubit Hermitian operator,
$U_j$ denotes the $j^{\rm th}$ two-qubit unitary,
and $H$ denotes an $n$-qubit Hermitian operator.}
\label{fig_Chart}
\end{figure}

The first step is a chart, a diffeomorphism that maps one manifold to another invertibly.
Our chart acts on the algebra 
$\mathfrak{su}(4)^{\times R}$ that generates $\mathrm{SU}(4)^{\times R}$.
To define the chart, we parameterize an element $H$ of 
the $j^{\rm th}$ copy of $\mathfrak{su}(4)$: 
\begin{equation}
	H=\sum_{\substack{
	\alpha, \beta \in \{\mathbb{1},X,Y,Z\} \\
	(\alpha, \beta)\neq (\mathbb{1},\mathbb{1})}} 
	\lambda_{j, \alpha, \beta} \;
	\alpha \otimes \beta,
\end{equation}
wherein $\lambda_{j, \alpha, \beta}\in \rr$.
For each point $x= (U_1, U_2, \dots ,U_{R}) \in \mathrm{SU}(4)^{\times R}$, 
we define the local exponential chart 
$\exp^{\times R}_x:\mathfrak{su}(4)^{\times R}\to	
\mathrm{SU}(4)^{\times R}$ 
as 
$\exp^{\times R}_x(H_1,\dots ,H_{R})
:=(e^{\mathrm{i} H_1}U_1 ,\dots ,e^{\mathrm{i} H_{R}}U_{R})$,
and we define the analogous 
$\exp_U:\mathfrak{su}(2^n)\to \mathrm{SU}(2^n)$ as 
$\exp_U(H):= e^{\mathrm{i}H}U$.
These charts are standard for matrix Lie groups.
Both are locally invertible in small neighbourhoods around $x$ and $U$, by a standard result in Lie-group theory~\cite{hall2015lie}.
The three-part map, represented by the dashed lines in Fig.~\ref{fig_Chart}, has the form 
$\exp_{F^A(x)}^{-1} \circ F^A \circ \mathrm{exp}_x^{\times R}$.

We now characterize the map's derivative,
to characterize the derivative of $F^A$,
to characterize the rank of $F^A$.
Denote by $D_0$ the derivative evaluated where
the Hermitian operators are set to zero,
such that each chart reduces to the identity operation.
The image of $D_0\LParen
\exp_{F^A(x)}^{-1} \circ F^A \circ \mathrm{exp}_x^{\times R}
\RParen$ is spanned by the operators 
\begin{equation}
\partial_{\lambda_{j,A,B}} 
\left(\exp^{-1}_{F^A(x)}\circ F^A\circ \exp_x^{\times R} \right)\Big|_0.
\end{equation}
These operators have the form
\begin{align}
   \label{eq_Deriv_Form}
   U_R \ldots U_{j+1} P U_{j} \ldots U_{1} ,
\end{align}
wherein $P$ denotes a two-qubit Pauli operator.
We apply the setting above to prove the following lemma.

\locus*

\begin{proof}
Consider representing an operator~\eqref{eq_Deriv_Form}
as a matrix relative to an arbitrary tensor-product basis.
To identify the matrix's form, we imagine representing
the unitaries in $\mathrm{SU}(4)^{\times R}$ as matrices
relative to the corresponding tensor-product basis for $\mathbb{C}^{2}\otimes \mathbb{C}^{2}$.
Combining the unitary matrices' elements polynomially
yields the matrix elements of~\eqref{eq_Deriv_Form}.

$D_x F^A$ has the same rank as
$D_0 ( \exp_{F^A(x)}^{-1} \circ F^A \circ
\mathrm{exp}_x^{\times R})$, because
$\exp_x^{\times R}$ and $\exp_{F^A(x)}$ are local charts~\cite{lee2013smooth}.
Recall that $E_{<r_{\mathrm{max}}}$ denotes the locus of points, 
in $\mathrm{SU}(4)^{\times R}$, 
where $F^A$ has a rank $<r_{\mathrm{max}}$.
Equivalently, by the invertible-matrix theorem,
$E_{<r_{\mathrm{max}}}$ consists of the points 
where certain minors of 
$D_0( \exp^{-1}\circ F^A \circ 
\mathrm{exp}_x^{\times R})$---the 
determinants of certain collections of
$r_{\mathrm{max}}\times r_{\mathrm{max}}$ matrix elements---vanish. 
The determinants' vanishing implies 
a set of equations polynomial in
the matrix elements of 
$D_0( \exp^{-1}\circ F^A \circ 
\mathrm{exp}_x^{\times R})$---and so,
by the last paragraph, polynomial in
the entries of matrices in $\mathrm{SU}(4)^{\times R}$.          
$\mathrm{SU}(4)^{\times R}$ is a real algebraic set, being the set of operators that satisfy the polynomial equations equivalent to $UU^{\dagger}=\mathbb{1}$ and $\det U=1$.
Thus, by Definition~\ref{def_Alg_Subset},
the points of rank $<r$ form an algebraic subset of $\mathrm{SU}(4)^{\times R}$.

We can now invoke properties of algebraic subsets,
reviewed in Appendix~\ref{app_Alg_Review}.
First, we prove that $\mathrm{SU}(4)^{\times R}$
is irreducible in the Zariski topology.
The Zariski topology of $\mathrm{SU}(4)^{\times R}$ is coarser than the topology inherited from $(\mathbb{C}^{4\times 4})^{\times R}$, 
identified with $\rr^{32R}$. 
As $\mathrm{SU}(4)^{\times R}$ is connected in the finer topology, 
so is $\mathrm{SU}(4)^{\times R}$ connected in the Zariski topology.
This connectedness implies that 
$\mathrm{SU}(4)^{\times R}$ is irreducible,
as $\mathrm{SU}(4)^{\times R}$ is an algebraic group~\cite[Summary~1.36]{milne2017algebraic}.
Being irreducible, $\mathrm{SU}(4)^{\times R}$ has a dimension
\`{a} la Definition~\ref{def:dimension}. 
If the low-rank locus $E_{<r_{\mathrm{max}}}$ is not all of $\mathrm{SU}(4)^{\times R}$, then it is, by Definition~\ref{def:dimension}, a lower-dimensional algebraic subset.
Every dimension-$N$ algebraic subset decomposes into 
a collection of submanifolds, 
each of which has dimension $\leq N$~\cite[Prop.~9.1.8]{bochnak2013real}.
As a proper submanifold has measure $0$, 
$E_{<r_{\mathrm{max}}}$ has measure 0.
As an algebraic subset, $E_{<r_{\mathrm{max}}}$ is closed in the Lie-group topology. 

Finally, we prove that $d_{A}=r_{\mathrm{max}}$.
In a small open neighborhood $V$ of a point 
$x\in E_{ r_{\mathrm{max}} }$, 
the contraction map's rank is constant, by
Lemma~\ref{lemma:locus}.
By the constant-rank theorem~\cite[Thm~5.13]{lee2013smooth}, therefore,
$F^{A_T}$ acts locally as a projector throughout $E_{r_{\mathrm{max}}}$---and
so throughout $\mathrm{SU}(4)^{\times R}$
(except on a region of measure 0,
by Lemma~\ref{lemma:locus}).
The projector has a rank, like $F^{A_T}$, of $r_{\mathrm{max}}$.
A rank-$r_{\mathrm{max}}$ projector has an image that is 
a dimension-$r_{\mathrm{max}}$ manifold.
Hence $r_{\mathrm{max}}\leq d_A$.
The other direction, $d_A\leq r_{\mathrm{max}}$, follows directly from Sard's theorem~\cite{sard1965hausdorff}.
Let $X_r$ denote the set of points where $F^A$ is rank-$r$.
As $F^A$ is a smooth map, Sard's theorem ensures that 
$r$ upper-bounds the Hausdorff dimension of the image $F^A(X_r)$.
As $F^A \LParen \mathrm{SU}(4) \RParen$ 
is a semialgebraic set, it stratifies into manifolds, by Theorem~\ref{thm_Stratification}.
Therefore, the Hausdorff dimension coincides with the semialgebraic set's dimension.
\end{proof}

Lemma~\ref{lemma:locus}, combined with the following lemma, implies Proposition~\ref{thm:lower-bound-on-dA}.

\begin{restatable}[Existence of a high-rank point]{lemma}{pointofhighrank}
	\label{lemma:pointofhighrank}
	Let $T\in \mathbb{Z}_{> 0}$ denote any nonnegative integer.
	Consider any architecture $A_T$ formed from
	$T$ \new{$L$-gate, backwards-light-cone--containing blocks}.
	The map $F^{A_{T}}$ has the greatest rank possible, 
	$r_{\mathrm{max}} \geq T$.
\end{restatable}
	
	\begin{figure}
	\begin{center}
 	\includegraphics[width=0.8\columnwidth]{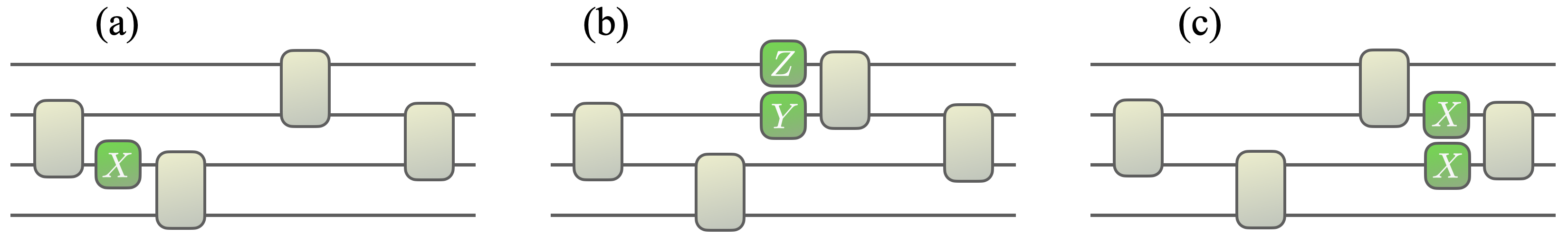}
 	\end{center}
	\caption{Examples of partial derivatives 
	$\partial_{\lambda_{j,\alpha,\beta}} ( 
	\exp^{-1}_{F^A(x)}\circ F^A \circ \exp_x^{\times 4} ) 
	\big\rvert_{ \lambda_{j,\alpha,\beta} = 0}$ 
	that span the image of 
	$D_0( \exp_{F^A(x)} \circ F^A\circ \exp_x^{\times 4})$.}
	\label{figure:partialderivatives}
\end{figure}

\begin{proof}
Without loss of generality, we assume that all $T$ \new{blocks} have identical architectures. This assumption will simplify the notation below.
We can lift the assumption by complicating the notation.
	
Consider an arbitrary point $x=(U_1, U_2, \ldots,U_R)\in \mathrm{SU}(4)^{\times R}$.
For all $x$, the contraction map $F^{A_T}$ has a derivative characterized, in the proof of Lemma~\ref{lemma:locus},
with local charts $\exp_{F^{A_T}(x)}$ and $\exp^{\times R}_x$.
The number of gates in $A_T$ is $R \leq TL$.
The map $F^{A_T}$ has an image spanned by the partial derivatives 
$\partial_{\lambda_{j, \alpha, \beta}}
( \exp^{-1}_{F^{A_T}(x)}\circ F^{A_T}\circ\exp^{\times R}_x)
\big\lvert_{ \lambda_{j, \alpha, \beta} = 0 }$.
Each partial derivative has the form 
\begin{align}
	\label{eq_Lemma2_Help0}
	U_R U_{R-1} \ldots U_{j+1} (\alpha \otimes \beta) 
	U_{j} U_{j-2} \ldots U_1
\end{align}
(Fig.~\ref{figure:partialderivatives}).
$\alpha$ and $\beta$ denote Pauli operators; each acts nontrivially on just one of the two qubits on which $U_j$ acts nontrivially.
We implicitly pad operators with identities wherever necessary,
such that the operators act on the appropriate Hilbert space.
	
We aim to lower-bound the greatest possible rank, $r_{\mathrm{max}}$, of the map $F^{A_T}$.
To do so, we construct a point 
\begin{align}
	x_T = \left(
	\underbrace{C_1^{(1)},\ldots,C_1^{(L)}}_{L~\text{gates}},
	\ldots, \underbrace{C_T^{(1)},\ldots, C_T^{(L)}}_{L~\text{gates}} \right)
	\in \mathrm{SU}(4)^{\times R} \ .
\end{align}
We will choose for the $C_j^{(i)}$'s to be Clifford gates.
A gate's subscript, $j$, labels the \new{blocks} to which the gate belongs.
The superscript, $i$, labels the gate's position within the \new{block}.
The gates constitute a \new{block} as
$C^{(L)}_j C_j^{(L-1)} \ldots C^{(1)}_j =: C_j$.
Our construction of $C_j$ relies on a property of
an arbitrary Pauli operator $Q_j$:
We can choose the Clifford gates $C^{(i)}_j$ such that \new{block} $C_j$ maps $Q_j$ to a $Z$ on qubit $t$: 
$C_j Q_j C_j^{\dagger} = Z_{t}
\equiv \mathbb{1}^{\otimes (t-1)} \otimes Z \otimes \mathbb{1}^{\otimes (n-t)}$.
We now show how the existence of such a Clifford unitary $C_j$ implies Lemma~\ref{lemma:pointofhighrank}. 
Afterward, we show to construct $C_j$.

Let us choose the Pauli strings $Q_j$ 
that guide our construction of the Clifford \new{block} $C_j$.
We choose the $Q_j$'s inductively over $T$ such that 
$\{ (C_T C_{T-1} \ldots C_j) Q_j 
(C_{j-1} C_{j-1} \ldots C_1) \}_{1\leq j\leq T}$ 
is linearly independent. 
We start with an arbitrary Pauli string $Q_1$.
The form of $Q_1$ guides our construction of $C_1$.
Second, we choose for $Q_2$ to be an arbitrary Pauli string 
$\neq C_1 Q_1 C_1^{\dagger}$.
$Q_2$ guides our construction of $C_2$.
Third, we choose for $Q_3$ to be an arbitrary Pauli string outside
$\mathrm{span} \{ C_1C_2Q_1C_2^{\dagger}C_1^{\dagger},
    C_{2}Q_2C_2^{\dagger} \}$.
This $Q_3$ guides our construction of $C_3$.
After $T$ steps, we have constructed all the $Q_j$'s and $C_j$'s.
If $T < 4^n-1$, enough Pauli strings exist that, at each step, a Pauli string lies outside the relevant span.

The operators $(C_T C_{T-1} \ldots C_j) Q_j (C_{j-1} C_{j-2} \ldots C_1)$, 
for $j \in [1, T]$, 
are in the image of 
$D_0 ( \exp^{-1}_{F^{A_{T}}(x_{T})}
\circ  F^{A_{T}}\circ 
\exp_{x_{T}}^{\times (R)} )$:
\begin{align}
\begin{split}\label{eq_partial_extended}
\partial_{\lambda_{jL,\mathbb{1}_1,\mathbb{Z}}} 
\left( \exp^{-1}_{F^{A_{T}}(x_{T})}
\circ  F^{A_{T}} \circ \,
\exp_{x_{T}}^{\times R} \right)\Big|_0 \quad 
&=(C_{T} C_{T-1} \ldots C_{j+1})
(\mathbb{1}_{t-1}\otimes Z_{t}\otimes \mathbb{1}_{n-t}) 
(C_j C_{j-1} \ldots C_1) \\
&=(C_{T} C_{T-1} \ldots C_{j}) 
Q_j (C_{j-1} C_{j-2} \ldots C_1).
\end{split}
\end{align}
We have assumed, without loss of generality, that each \new{block}'s final gate acts on qubit $t$.
For all $j \in [1, T]$, the operators 
$(C_T C_{T-1} \ldots C_j) Q_j 
(C_{j-1} C_{j-2} \ldots C_1)$
are in the image of 
$D_0 ( \exp^{-1}_{F^{A_{T}}(x_{T})}
\circ  F^{A_{T}}\circ 
\exp_{x_{T}}^{\times R} )$ and are linearly independent.
Therefore, the rank of $F^{A_T}$ at the point $x_T$ is $\geq T$.

In the remainder of this proof, we provide the missing link: We show that, for every Pauli string $P$, we can construct a \new{backwards-light-cone--containing block} that implements a Clifford unitary 
$C = C^{(L)} C^{(L-1)} \ldots C^{(1)}$ such that 
$C P C^{\dagger}=Z_t$.
We drop subscripts because subscripts index \new{blocks} and
this prescription underlies all \new{blocks}.
By definition, \new{each block contains} 
a qubit $t$ to which each other qubit $t'$ connects
via gates in the \new{block}.
The path from a given qubit $t'$ depends on $t'$,
and multiple paths may connect a $t'$ to $t$.
Also, one path may connect $t$ to multiple qubits.
We choose an arbitrary complete set of paths (which connect all the other qubits to $t$) that satisfies the merging property described below. To introduce the merging property, we denote by $m$ the number of paths in the set.
Let $p\in [1,m]$ index the paths.
Path $p$ contacts the qubits in the order
$i_{p,1} \mapsto i_{p,2} \mapsto \ldots \mapsto i_{p,l_p}=t$,
reaching $l_p \in [1, \, L+1]$ qubits.
We choose the paths such that they merge whenever they cross: 
If $i_{p,j}=i_{p',j'}$, then $i_{p,j+k}=i_{p',j'+k}$ for all 
$k \in \{1, 2, \ldots, l_p-j{=}l_{p'}-j' \}$.
We choose for all the gates outside these paths to be identities.
Next, we choose the nontrivial gates in terms of 
an arbitrary Pauli string.

Let $P = \bigotimes_{j=1}^n P_j$ denote an arbitrary nontrivial $n$-qubit Pauli string. 
Some Clifford unitary $C$ maps $P$ to 
a Pauli string that acts nontrivially on just one qubit
(see Refs.~\cite{cleve_near-linear_2015,Webb3Design,PhysRevA.96.062336} and~\cite{footnoteclifford}),
%\footnote{
%	Consider conjugating an arbitrary $n$-qubit Pauli operator $P$ with a uniformly random Clifford operator $C$.
%	The result, $C^\dag P C$, is a uniformly random $n$-qubit Pauli operator~\cite{cleve_near-linear_2015,Webb3Design,PhysRevA.96.062336}.
%	Therefore, for every initial Pauli operator $P$ and every final Pauli operator, some Clifford operator $C$ maps one to the other.}),
which we choose to be $t$.
We arbitrarily choose for the string's nontrivial single-qubit Pauli operator to be $Z$.
Let $i_{p,k_p}$ denote the first index $j$ in
$i_{p,1} \mapsto i_{p,2} \mapsto \ldots \mapsto i_{p,l_p}$ for which 
$P_j \neq \mathbb{1}$.
By definition, $P_{i_{p,k}}\otimes P_{i_{p,k+1}}$ 
is a nontrivial Pauli string.
There exists a two-local Clifford gate $C^{p,(0)}$ 
that transforms $P_{i_{p,k}}\otimes P_{i_{p,k+1}}$ into
a $Z$ acting on qubit $i_{p,k+1}$:
\begin{equation}
C^{p,(0)} (P_{p,i_k}\otimes P_{p,i_{k+1}}) (C^{p,(0)})^{\dagger}
= \mathbb{1}_{i_k} \otimes Z_{i_{k+1}}.
\end{equation}
Operating with $C^{p,(0)}$ (padded with $\mathbb{1}$'s) 
on the whole string $P$ 
yields another Pauli string:
\begin{equation}
C^{p,(0)} P (C^{p,(0)})^{\dagger}
=\bigotimes_{j=1}^n P^{p,(1)}_j.
\end{equation}
Let $i_{p,\ell}$ denote the first index $j$ for which
$P^{p,(1)}_{j}$ is a nontrivial Pauli operator.
Since 
\begin{equation}
P^{(1)}_{i_{p,k+1}} \otimes 
P^{(1)}_{i_{p,k+2}} = Z_{i_{p,k+1}} \otimes P_{i_{p,k+2}},
\end{equation}
$i_{p,\ell} = i_{p,k+1}$.
There exists a two-local Clifford gate $C^{p,(1)}$ that
shifts the $Z$ down the path:
\begin{equation}
C^{p,(1)} \left( 
P^{p,(1)}_{i_{p,k+1}} \otimes P^{p,(1)}_{i_{p,k+2}} \right) 
(C^{p,(1)})^{\dagger}
= \mathbb{1}_1 \otimes Z_{i_{p,k+2}}.
\end{equation}
We perform this process---of
shifting the $Z$ down the path and 
leaving an $\mathbb{1}_1$ behind---for every path simultaneously.
For example, if we begin with two equal-length paths, 
$C^{(2)} = C^{p,(2)} C^{p',(2)}$. 
This simultaneity is achievable until two paths merge.
Whenever paths merge, we choose the next Clifford gate such that we proceed along the merged path.
Every qubit is visited, and every path ends at qubit $t$.
Therefore, we have constructed a circuit that implements
a Clifford operation $C$ such that $C PC^{\dagger}= Z_t$.
Figure~\ref{fig:PerturbationCliffordCircuits}(d) depicts an example of this construction.
\end{proof}

The foregoing proof has a surprising implication:
A map's rank is somewhat divorced from a circuit's complexity.
The rank of $F^{A_T}$ at $x_T$ is at least $T$, which could be a large number.
Yet,  the contracted unitary corresponding to this circuit is Clifford.
Hence the extended circuit's complexity surpassed the original circuit's complexity only a little---by, at most,
$O \LParen n^2/\log(n) \RParen$~\cite{aaronson2004improved}.

Finally, we combine Lemmata~\ref{lemma:locus} and~\ref{lemma:pointofhighrank} to prove Theorem~\ref{theorem:linear growth}:
\begin{figure}
	\begin{center}
	\includegraphics[width=0.4\columnwidth]{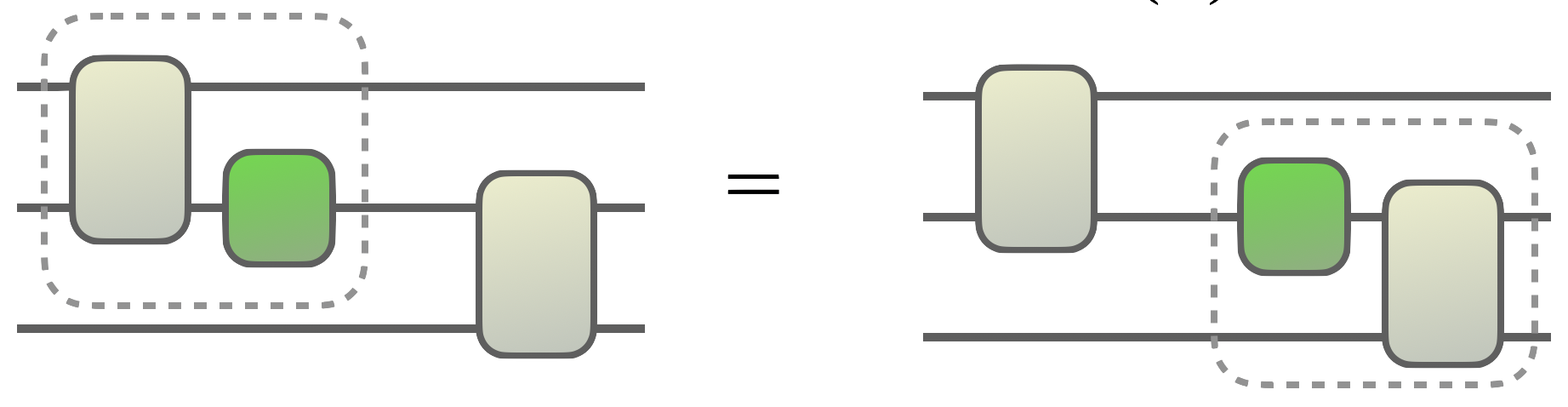}
	\end{center}
\caption{Along every contraction is a redundant copy of the gauge group $\mathrm{SU}(2)$.
}
\label{figure_gauge}
\end{figure}

\lineargrowth*

\begin{proof}[Proof of Theorem~\ref{theorem:linear growth}]
We reuse the notation introduced in Lemmata~\ref{lemma:locus} and~\ref{lemma:pointofhighrank}.
Examples include $A_T$, an arbitrary architecture
that satisfies the assumptions in Lemma~\ref{lemma:pointofhighrank} 
and that consists of $R {\leq} TL$ gates.
$F^{A_T}$ denotes the corresponding contraction map.
$E_{r_{\rm max}}$ denotes the locus of points at which
$F^{A_T}$ achieves its greatest rank, $r_{\mathrm{max}}$.
In a small open neighborhood $V$ of a point 
$x\in E_{ r_{\mathrm{max}} }$, 
the contraction map's rank is constant, by
Lemma~\ref{lemma:locus}.
By 
the constant-rank theorem~\cite[Thm~5.13]{lee2013smooth}, therefore,
$F^{A_T}$ acts locally as a projector throughout $E_{r_{\mathrm{max}}}$---and
so throughout $\mathrm{SU}(4)^{\times R}$
(except on a region of measure 0,
by Lemma~\ref{lemma:locus}).
The projector has a rank, like $F^{A_T}$, of $r_{\mathrm{max}}$.
Therefore, in the open set 
$V \subseteq \mathrm{SU}(4)^{\times R}$, 
$F^{A_T}$ is equivalent, up to a diffeomorphism, to the projection
\begin{equation}
	\label{eq_Project_Equiv}
	\left( x_1,\dots ,x_{\dim \LParen \mathrm{SU}(4)^{\times R} \RParen} \right)
	\mapsto (x_1,\dots ,x_{r_{\rm max}},
	\underbrace{0, \ldots, 0}_{
	            \dim \LParen \mathrm{SU}(2^n) \RParen - r_{\mathrm{max}} } ).
\end{equation}
For simplicity of notation, we identify $V$ with its image under the local diffeomorphism 
(we do not distinguish $V$ from its image notationally).

The open subset $V$ contains, itself, an open subset
that decomposes as a product:
$V_1\times V_2\subseteq V$, such that
$x\in V_1\times V_2$ and,
as suggested by Eq.~\eqref{eq_Project_Equiv},
\begin{align}
   \label{eq_V_Dims}
   V_1\subseteq \rr^{r_{\rm max}},
   \quad \text{and} \quad
   V_2 \subseteq \rr^{\dim \LParen \mathrm{SU}(4)^{\times R} \RParen
- r_{\rm max} } .
\end{align}
(Again to simplify notation, we are equating the local sets $V_{j=1,2}$ with their images, under local charts, in $\rr^m$, 
for $m \in \mathbb{Z}_{>0}$.)
From now on, $V_1 \times V_2$ is the open subset of interest.
The contraction map's equivalence to a projector,
in $V_1 \times V_2$,
will help us compare high-depth circuits with low-depth circuits:
Consider a circuit whose contraction map
takes some local neighborhood to an image of some dimension.
How does the dimension differ between
high-depth circuits and low-depth circuits?
We start by upper-bounding the dimension for low-depth circuits.

We have been discussing an $R$-gate architecture $A_T$.
Consider any smaller architecture $A'$ of $R' < R$ gates.
$A'$ is encoded in a contraction map $F^{A'}$ 
whose domain is $\mathrm{SU}(4)^{\times R'}$.
As explained in the proof of Lemma~\ref{lemma:locus}, 
$F^{A'}$ is a polynomial map. 
Therefore, $F^{A'}$ has a property prescribed by
the Tarski-Seidenberg principle~\cite{bochnak2013real} (Theorem~\ref{theorem:Tarski}):
The image 
$F^{A'} \LParen \mathrm{SU}(4)^{\times R'}  \RParen$ 
is a semialgebraic set of dimension 
$\leq \dim \LParen \mathrm{SU}(4)^{\times R'} \RParen
= R'\dim \LParen {\rm SU}(4) \RParen = 15R'$.
	
We can strengthen this bound:
Consider contracting two gates that share a qubit.
The shared qubit may undergo a one-qubit gate
specified by three parameters (one parameter per one-qubit Pauli).
The one-qubit gate can serve  as 
part of the first two-qubit gate or 
as part of the second two-qubit gate;
which does not affect the contraction.
Hence the contraction contains 
3 fewer parameters than expected \footnote{
		In other words, the contraction has a redundant copy of the gauge group SU$(2)$:
		Every unitary $U\in$ SU$(4)$ decomposes as 
		$(U_1\otimes U_2) K(V_1\otimes V_2)$, wherein 
		$K=e^{\mathrm{i}(aZ\otimes Z+bY\otimes  Y+cX\otimes X)}$ 
		and the $U_j$ and the $V_j$ denote single-qubit unitaries~\cite{khaneja2000cartan}.}.
Let us classify the shared qubit as 
an input of the second two-qubit gate.
A two-qubit gate in a circuit's bulk accepts two input qubits
outputted by earlier gates.
So we might expect an $R'$-gate circuit to have
$\dim \LParen F^{A'}(\mathrm{SU}(4)^{\times R'}) \RParen
\leq 15R' - 2 \times 3 R' = 9 R'$.
But the first $n/2$ gates
[the leftmost vertical line of gates in Fig.~\ref{figure:brickwork}(a)]
receive their input qubits from no earlier gates.
So we must restore $3 \times 2$ parameters for each of the $n/2$ initial gates, or restore $3n$ parameters total \footnote{
	Technically, this bound on the dimension does not follow from Lemma~2 % in the Supplementary Material 
	as it does for the bound 
	$\dim F^{A'}(\mathrm{SU}(4)^{\times R'})\leq 15 R'$. 
	The reason is, the quotient space 
	$\mathrm{SU}(4)^{\times R'}/SU(2)^{\times (2R'-n)}$ is not necessarily semialgebraic.
	This difficulty can be resolved via Sard's theorem~\cite{sard1965hausdorff}, which asserts, as a special case, that the Hausdorff dimension of a smooth map's image is bounded by its domain's dimension. 
	A semialgebraic set's dimension is the greatest dimension in its stratification and so agrees with the Hausdorff dimension.
}:
\begin{align}
   \label{eq_Upperbound_Lowdepth}
   \dim \LParen F^{A'}(\mathrm{SU}(4)^{\times R'}) \RParen
   \leq  9R'+3n .
\end{align}
We have upper-bounded the dimension for low-depth circuits.

We now lower-bound the corresponding dimension
for high-depth circuits.
We can do so by lower-bounding the greatest possible rank, $r_{\mathrm{max}}$, 
of a high-depth architecture's contraction map, $F^{A_T}$:
In an open neighborhood of $x \in \mathrm{SU}(4)^{\times R}$,
$F^{A_T}$ is equivalent to a projector, which has some rank.
The neighborhood's image, under the projector,
is a manifold. The manifold's dimension equals the projector's rank.
Therefore, we bound the rank to bound the dimension.

Augmenting an architecture with $T$ $(\leq L)$-gate \new{blocks}
increases the contraction map's greatest possible rank, $r_{\mathrm{max}}$,
by $\geq T-1$.
Therefore, for an architecture-$A_T$ circuit of $R \leq TL$ gates,
we have constructed a point of rank 
$ T \geq R/L$. Therefore, 
\begin{align}
   \label{eq_RMax_Bound}
   r_{\mathrm{max}} \geq R  / L.
\end{align}

We have lower-bounded the dimension of
the image of a high-depth architecture's contraction map
[the rank in Ineq.~\eqref{eq_RMax_Bound}]
and have upper-bounded the analogous dimension for 
a low-depth architecture
[Ineq.~\eqref{eq_Upperbound_Lowdepth}].
The high-depth-architecture dimension
upper-bounds the low-depth-architecture dimension,
\begin{equation}
    \label{eq:comparedimensions}
	\dim \big( F^{A'} \LParen 
	\mathrm{SU}(4)^{\times R'} \RParen \big) < r_{\mathrm{max}} ,
\end{equation}
if
\begin{align}
   \label{eq_RPrime_Bound1}
   9R' + 3n < r_{\mathrm{max}} ,
\end{align} 
by Ineq.~\eqref{eq_Upperbound_Lowdepth}.
Furthermore, by Ineqs.~\eqref{eq_Upperbound_Lowdepth}
and~\eqref{eq_RMax_Bound},
Ineq.~\eqref{eq:comparedimensions} holds if
$9 R' + 3n < R /L$, or 
	$R' < \frac{R}{9L} - \frac{n}{3} \, .$
\begin{equation}
   \label{eq_gatebound2}
  R' < \frac{R}{9L} - \frac{n}{3} \, .
\end{equation}
holds. We have upper-bounded the short circuit's gate count
in terms of the deep circuit's gate count.

Let us show that, if Ineq.~\eqref{eq_gatebound2} holds,
the short circuits form 
a set of measure 0 in SU$(4)^{\times R}$ \footnote{
Technically, the \emph{sets of few gates} form 
a set of measure 0 in 
SU$(4)^{\times R}$.
Circuits are not in SU$(4)^{\times R}$, 
as explained in Fig.~\ref{fig_Contraction_map}.
However, we expected the sentence above to be more intuitive with ``short circuits'' instead of ``sets of few gates.''}.
We will begin with a point 
$x\in E_{r_{\mathrm{max}}}$;
apply the short-architecture contraction map $F^{A'}$;
and follow with the deep-architecture contraction map's inverse, $(F^{A_T})^{-1}$.
The result takes up little space in SU$(4)^{\times R}$,
we will see.

To make this argument rigorous, we recall 
the small open neighborhood $V_1 \times V_2$
of $x\in E_{r_{\mathrm{max}}}$.
In $V_1\times V_2$, 
$F^{A'} \LParen \mathrm{SU}(4)^{\times R'} \RParen$
has the preimage, under $F^{A_T}$, of
\begin{equation}
	\left( F^{ A _T}|_{V_1\times V_2} \right)^{-1} 
	\left( F^{A'} \LParen \mathrm{SU}(4)^{\times R'} \RParen
	\right)
	\simeq
	\left[ F^{A'} \LParen \mathrm{SU}(4)^{\times R'} \RParen 
	\cap V_1 \right] \times V_2. 
\end{equation}
The $\simeq$ represents our identification of the map $F^A$ with its representation in local charts.
By the proof of Lemma~\ref{lemma:locus}, 
$F^{A'} \LParen \mathrm{SU}(4)^{\times R'} \RParen$ 
is a semialgebraic set. Therefore,
by Theorem~\ref{thm_Stratification},
$F^{A'} \LParen \mathrm{SU}(4)^{\times R'} \RParen$ 
is a union of smooth manifolds. 
Each manifold is of dimension $\leq 9R'+3n$, 
by Theorem~\ref{thm_Stratification} and  Ineq.~\eqref{eq_Upperbound_Lowdepth}.
By Eq.~\eqref{eq_V_Dims}, $V_2$ is of dimension 
$\dim \LParen \mathrm{SU}(4)^{\times R} \RParen
 - r_{\mathrm{max}}$. Therefore, 
$[F^{A'} \LParen \mathrm{SU}(4)^{\times R'} \RParen \cap V_1 ]
 \times V_2$ consists of manifolds of dimension 
$\leq 9 R'+3n
+ \dim \LParen \mathrm{SU}(4)^{\times R} \RParen 
- r_{\rm max}$.
Using Ineq.~\eqref{eq_RPrime_Bound1},
we can cancel the $9 R'+3n$ with the $- r_{\rm max}$,
at the cost of loosening the bound:
$[F^{A'} \LParen \mathrm{SU}(4)^{\times R'} \RParen \cap V_1 ]
 \times V_2$ consists of manifolds of dimension
$< \dim \LParen \mathrm{SU}(4)^{\times R} \RParen$.
As a collection of manifolds of submaximal dimension, 
the unitaries implemented by short circuits satisfying~\eqref{eq_gatebound2},
restricted to a small open neighborhood $V_1$, 
form a set of measure $0$ \footnote{Again, by ``short circuits,'' we mean, ``sets of few gates.''
We replaced the latter phrase for ease of expression.}.

Let us extend this conclusion about $n$-qubit unitaries---about images of maps $F^{A'}$---to a conclusion about preimages---about lists of gates.
By Lemma~\ref{lemma:locus}, $E_{r_{\mathrm{max}}}$ is of measure $1$.
Therefore, for every $\varepsilon>0$,
there exists a compact subset $K\subseteq E_{r_{\mathrm{max}}}$ 
of measure $1-\varepsilon$.
Since $K$ is compact, for any cover of $K$ by open subsets, 
a finite subcover exists. 
The foregoing paragraph shows that, 
restricted to each open set in this finite subcover, 
the preimage of the unitaries reached by lower-depth circuits is of measure $0$.
Therefore, the preimage of the $R'$-gate, architecture-$A'$ circuits 
is of measure $\leq \varepsilon$. 
Since $\varepsilon > 0$ is arbitrary, 
the preimage is of measure $0$.
The foregoing argument holds for each architecture $A'$ of $R'$ gates. 
Hence each preimage forms a set of measure $0$. 
The total measure is subadditive.
So the union of the preimages, over all architectures with $\leq R'$ gates, is of measure $0$.
We have proven the circuit-complexity claim posited in Theorem~\ref{theorem:linear growth}. 
The state-complexity claim follows from tweaks to the proof
(Appendix~\ref{sec:state-complexity}).
\end{proof}

\section{Proof of the linear growth of state complexity} \label{sec:state-complexity}

At the end of Appendix~\ref{app_Prove_Lemma1},
we proven part of Theorem~\ref{theorem:linear growth}---that
circuit complexity grows linearly with the number of gates.
Here, we prove rest of the theorem---that
state complexity grows linearly. 
We need only tweak the proof presented in Appendix~\ref{app_Prove_Lemma1}.

Consider instead of the contraction map $F^{A_T}$, 
the map that contracts a list of gates, 
forming an architecture-$A_T$ circuit, 
and applies the circuit to $|0^n\rangle$, to get
\begin{equation}
G^{A_T} :\mathrm{SU}(4)^{\times R}
\to S^{2 \times2^n-1}
\subseteq \mathbb{C}^{2^n} .
\end{equation}
The argument works the same as in Appendix~\ref{app_Prove_Lemma1}, with one exception:
The derivative $D_x G^{A_T}$ has an image that does not contain
$4^n-1$ nontrivial linearly independent Pauli operators.
Rather, the image contains the computational basis $\{\mathrm{i}^{\kappa}|x\rangle\}_{x\in\{0,1\}^n,\kappa\in\{0,1\}}$ 
formed by applying tensor products of $Z, X$ and $Y$
to $\ket{0^n}$.
(We denote the imaginary number $\sqrt{-1}$ by $\mathrm{i}$.)
The proof of Lemma~\ref{lemma:locus} ports over without modification, as $G^{A_T}$ is a polynomial map between algebraic sets. 

The proof of Lemma~\ref{lemma:pointofhighrank} changes slightly.
We must prove the existence of a point
$x \in \mathrm{SU}(4)^{\times R}$ at which
$G^{A_T}$ has a rank at least linear in the circuit depth.
The only difference in the proof is, we must choose the operators $Q_j$ inductively such that the states 
$(C_T C_{T-1} \ldots C_j) Q_j (C_{j-1} C_{j-2} \ldots C_1)|0^n\rangle$ are linearly independent.
Such a choice is possible if $T< 2 \times 2^n-1$,
the number of real parameters in a pure $n$-qubit state vector.

\section{Randomized architectures}\label{sec_randomizedarchitectures}

From Theorem~\ref{theorem:linear growth} follows a bound on the complexity of a doubly random circuit:
Not only the gates, but also the gates' positions, are drawn randomly.
This model features in Ref.~\cite{brandao_local_2016}.
Our proof focuses on nearest-neighbor gates, but other models (such as all-to-all interactions) yield similar results.

\begin{restatable}[Randomized architectures]{cor}{randomarchitecture}\label{cor_randomizedarchitectures}
	Consider drawing an $n$-qubit unitary $U$ according to the following probability distribution:
	Choose a qubit $j$ uniformly randomly. 
	Apply a Haar-random two-qubit gate to qubits $j$ and $j+1$.
	Perform this process $R$ times.
	With high probability, the unitary implemented has a high complexity:
	For all $\alpha \in [0, 1)$,
	\begin{equation}
	\label{eq_Doubly_Random}
	\mathrm{Pr}\left(\compU(U)\geq \alpha \: \frac{R}{9n(n-1)^{2}}-\frac{n}{3}\right)\geq 1-\frac{1}{1-\alpha}(n-1)e^{-n} \, .
	\end{equation} 
\end{restatable}

\begin{proof}
	The proof relies on the following strategy:
	We consider constructing \new{blocks} randomly to form a circuit.
	If the \new{blocks} contain enough gates, we show,
	many of the \new{blocks contain backwards light cones}.
	This result enables us to apply Theorem~\ref{theorem:linear growth} to bound the circuit's complexity.
	
	Consider drawing $L$ gates' positions uniformly randomly.
	For each gate, the probability of drawing position $(j,j+1)$ is $1/(n-1)$. 
	The probability that no gates act at position $(j, j+1)$ is $(1-1/(n-1))^{L}$.
	Let us choose for each \new{block} to contain $L=n(n-1)^2$ gates. 
	Define a binary random variable $I_j$ as follows: 
	If one of the gates drawn during steps $(j-1)n(n-1),(j-1)n(n-1)+1,\ldots, jn(n-1)$ acts at $(j,j+1)$, then $I_j=1$. 
	Otherwise, $I_j=0$.
	With high probability, gates act at all positions:
	\begin{equation}
	p:=\mathrm{Pr}\left(\bigwedge_{j=1}^{n-1}(I_j=1)\right)=\left(1-\left(1-\frac{1}{n-1}\right)^{n(n-1)}\right)^{n-1}\geq \left(1-e^{-n}\right)^{n-1}\geq 1-(n-1)e^{-n}.
	\end{equation}
	We have invoked the inverse Bernoulli inequality and the Bernoulli inequality.
	We will use this inequality to characterize \new{blocks that contain backwards light cones}.
	
	Consider drawing $T$ $L$-gate \new{blocks} randomly,
	as described in the corollary.
	Denote by $X$ the number of \new{blocks} in which
	at least one position is bereft of gates:
	For some $j$, $I_j=0$.
	With high probability, $X$ is small:
	For all $a \in (0, T]$,
	\begin{equation}
	\mathrm{Pr}(X\geq a)
	\leq \frac{T(1-p)}{a}
	\leq T(n-1) \, e^{-L/n}/a \, ,
	\end{equation}
	by Markov's inequality.
	Let us choose for the threshold to be $a = (1 - \alpha)T$.
	With overwhelming probability, 
	$\alpha T$ \new{blocks} satisfy $\bigwedge_j (I_j=1)$ 
	and so contain gates that act at all positions $(j,j+1)$ in increasing order.
	Therefore, these \new{blocks} contain a staircase architecture and \new{so contain backwards light cones}.
	Therefore, a slight variation on Theorem~\ref{theorem:linear growth} governs the 
	$\alpha T \times L = \alpha R$ gates that form the \new{blocks}.
	Strictly speaking, Theorem~\ref{theorem:linear growth} governs only consecutive \new{backwards-light-cone--containing blocks}. In contrast, extra gates may separate the \new{blocks} here. However, the extra gates can only increase the contraction map's image. 
	Therefore, the additional $(1 - \alpha) TL$ gates cannot decrease the accessible dimension $d_{A_T}$. 
	Therefore, the bound from Theorem~\ref{theorem:linear growth} holds. With probability $\geq 1-\frac{1}{1-\alpha}(n-1) \, e^{-n}$ over the choice of architecture,
	\begin{equation}
	\compU(U)\geq \frac{R- (1 - \alpha) R}{9n^2(n-1)}-\frac{n}{3},
	\end{equation}
	with probability one over the choice of gates.
	This bound is equivalent to Ineq.~\eqref{eq_Doubly_Random}.
\end{proof}

\section{Proof of Corollary~\ref{cor:robustversion}}
\label{sec:proof_corollary}

Corollary~\ref{cor:robustversion} extends 
Theorem~\ref{theorem:linear growth}
to accommodate errors in the target unitary's implementation.
We prove Corollary~\ref{cor:robustversion} by 
drawing on the proof of Theorem~\ref{theorem:linear growth}
and reusing notation therein.

\begin{restatable}[Slightly robust circuit complexity]{cor}{robust} \label{cor:robustversion}
	Let $U$ denote the $n$-qubit unitary implemented by any random quantum circuit in any architecture $A_T$ that satisfies the assumptions in Theorem~\ref{theorem:linear growth}.
	Let $U'$ denote the $n$-qubit unitary implemented by any circuit of $R'\leq R/(9L)-n/3$ gates.
	For every $\delta \in (0, 1]$, there exists an $\varepsilon:=\varepsilon(A_T,\delta)>0$ 
	such that the Frobenius distance
	$d_{\rm F} (U, U') \geq \varepsilon$,
	with probability $1-\delta$, unless $R/L> 4^n-1$.
\end{restatable}
\begin{proof}[Proof of Corollary~\ref{cor:robustversion}]

The proof of Theorem~\ref{theorem:linear growth}
can be modified to show that, for every $\delta>0$, 
there exists an open set 
$B \subseteq \mathrm{SU} (2^n)$ that contains
$F^{A'} \LParen \mathrm{SU}(4)^{\times R'} \RParen$, 
such that the preimage
$(F^{A_T})^{-1}(B)$ is small---of measure $\leq \delta$.
The modification is as follows.
For every $\delta' > 0$, there exists 
a measure-$(1 - \delta')$ compact subset $K$ of $E_{r_{\mathrm{max}}}$.
As $K$ is compact, there exists a finite cover of $K$ 
that has the following properties:
$K$ is in the union $\cup_j V^j$ of subsets $V^j$.
On the $V^j$, the contraction map $F^{A_T}$ is equivalent to a projector, up to a local diffeomorphism.
As in the proof of Theorem~\ref{theorem:linear growth},
we can assume, without loss of generality, that
$V^j = V^j_1 \times V^j_2$.
The $V^j_1$ and $V^j_2$ are defined analogously to the $V_1$ and $V_2$ in the proof of Theorem~\ref{theorem:linear growth}.
For each $V^j$, there exists an open neighborhood $W^j$
of $F^{A'}(\mathrm{SU}(4)^{\times R'})\cap V^j_1$ 
such that $W^j$ has an arbitrarily small measure $\delta_j''>0$.
Therefore, $B:=\cup_j W^j$ has a preimage of measure
$\leq \delta'+\sum_j \delta''_j = \delta$.
Each of the summands, though positive, can be arbitrarily small.

The Frobenius norm induces a metric $d_{\mathrm{F}}$ on 
$\mathrm{SU}(4)^{\times R}$.
In terms of $d_{\mathrm{F}}$, we define the function 
\begin{equation}
	d_{\mathrm{F}} \left( \: . \: , \, 
	F^{A_T} \LParen \mathrm{SU}(4)^{\times R} \RParen 
	\setminus B \right):
	F^{A'} \LParen \mathrm{SU}(4)^{\times R'} \RParen
	\to \rr_{\geq 0}.
\end{equation}
This function is continuous, and 
$F^{A'} \LParen \mathrm{SU}(4)^{\times R'} \RParen$ is compact.
Therefore, the function achieves its infimum at a point $x_{\mathrm{min}}\in F^{A'} \LParen \mathrm{SU}(4)^{\times R'} \RParen$. 
Therefore, the minimal distance to
$F^{A_T} \LParen \mathrm{SU}(4)^{\times R} \RParen 
 \setminus B$ is
$d_{\mathrm{F}}
\LParen x_{\mathrm{min}}, F(\mathrm{SU}(4)^{\times R})\setminus B \RParen$.
Since $B$ is open, $F^A(\mathrm{SU}(4)^{\times R})\setminus B$ 
is closed and so compact.
By the same argument,
\begin{equation}
    \label{eq_vareps}
	\varepsilon (A_T,\delta) \coloneqq 
	d_{\mathrm{F}} \left( x_{\mathrm{min}}, 
	F \LParen \mathrm{SU}(4)^{\times R} \RParen 
	\setminus B \right)
	= \inf_{y\in 
	  F \LParen \mathrm{SU}(4)^{\times R} \RParen
	  \setminus B} 
	  \left\{ d_{\mathrm{F}} (x_{\mathrm{min}},y) \right\}
	=d_{\mathrm{F}} (x_{\mathrm{min}},y_{\mathrm{min}})>0.
\end{equation}
We have identified an $\varepsilon>0$ that satisfies Corollary~\ref{cor:robustversion}. 
\end{proof}

\section{Notions of circuit complexity}\label{sec:Nielsen}
As circuit complexity is a widely popular concept, there is a zoo of quantities that measure it.
We prove our main theorem for the straightforward definition of exact circuit implementation---the clearest and historically first notion of a circuit complexity---and for a version of approximate circuit complexity (Corollary~\ref{cor:robustversion}) with an uncontrollably small error. In this appendix, we briefly mention other notions of complexity, partially to
review other notions and partially to place the main text's findings in a wider context. Let $U\in$ SU$(2^n)$ denote a unitary. Ref.~\cite{nielsen2006quantum} discusses
notions of approximate circuit complexity.

\begin{definition}[Approximate circuit complexity]
The approximate circuit complexity $\compU(U,\eta)$ is the least number of $2$-local gates, arranged in any architecture, that implements $U$ up to
an error $\eta>0$ in operator norm $||.||$.
\end{definition}
\noindent
This definition is similar in mindset to the above (slightly) robust definition of a circuit complexity.
For every pair $U, U' \in$ SU$(2^n)$ of circuits, 
the Frobenius distance between them satisfies
\begin{equation}
	\frac{1}{2^n} d_{\rm F}(U, U')
	\leq || U-U'|| 
	\leq d_{\rm F}(U, U') .
\end{equation}
A widely used proxy for quantum circuit 
complexity---one that is increasingly seen
as a complexity measure in its own right---is Nielsen's geometric approach to circuit and
state complexity \cite{nielsen2005geometric,nielsen2006quantum,dowling2008geometry}.
This approach applies geometric reasoning to circuit complexity and led to many intuitive insights, including Brown and Susskind's conjectures about the circuit complexity's behavior under random evolution.
To connect to cost functions as considered in Nielsen's framework, 
consider \new{$1$-local} and $2$-local Hamiltonian terms $H_1, H_2, \dots, H_m$ 
in the Lie algebra $\text{su}(2^n)$ of traceless Hermitian matrices, normalized
as $\|H_j\|=1$ for $j=1,2,\dots, m$. 
Consider generating a given unitary, 
by means of a control system,
following Schr{\"o}dinger's equation:
\begin{equation}
	\frac{d}{dt}U(t) = -i H(t) U(t) ,
	\; \text{wherein} \;
	H(t) = \sum_{j=1}^m h_j (t) H_j .
\end{equation}
The control function $[0,\tau ]\rightarrow \rr^m$ is defined as
$t\mapsto \LParen h_1(t), \dots, h_m(t) \RParen$ and satisfies $U(0)=\mathbb{1}$. That is, a quantum circuit results from time-dependent control. In practice, not all of
$\rr^m$ reflects meaningful control parameters; merely
a control region $\mathcal{R} \subset \rr^m$ does. 
With each parameterized curve is associated a cost function
$c : \mathcal{R} \rightarrow \rr$, so that the entire cost of a unitary $U\in $SU$(2^n)$ becomes
\begin{equation}
C(U) := \inf_{T, \, t\mapsto H(t)}   
\int_{0}^\tau dt \; c \LParen H(t) \RParen .
\end{equation}
We take the infimum over all time intervals $[0, \tau ]$ and over all control functions $t\mapsto H(t)$
such that the control parameters are in $\mathcal{R}$ for all $t\in [0,\tau]$
and scuh that $U(\tau)=U$. 
Several cost functions are meaningful and have been discussed in the literature. A common choice is
\begin{equation}
c_p \LParen H(t) \RParen = \left(\sum_{j=1}^m h_j(t)^p\right)^{1/p}.
\end{equation}
In particular, $c_2$ gives rise to a sub-Riemannian metric. For the resulting cost $C_2(U)$, Ref.\ \cite{Nielsen_06_Optimal} establishes
a connection between the approximate circuit complexity and the cost:
Any bound on the approximate circuit complexity, with an approximation error bounded from below independently of the system size, immediately implies a lower bound on the cost.

\begin{theorem}[Approximate circuit complexity and cost \cite{Nielsen_06_Optimal}] For every integer $n$,
every $U\in \mathrm{SU}(2^n)$ and every $\eta>0$, 
\begin{equation}
\compU(U,\eta)\leq c \, \frac{C_2(U)^3  n^6}{\eta^2} \, .
\end{equation}
\end{theorem}
The quantity on the right-hand side can, in turn, be upper-bounded:
$C_2(U)\leq C_1(U)$.
This $C_1$ has a simple interpretation in terms of a weighted gate complexity~\cite{Entanglement}.

\begin{definition}[Weighted circuit complexities]
Let $U\in\mathrm{ SU}(2^n)$ denote a unitary.
The weighted circuit complexity ${\cal C}_{\mathrm{w}}(U)$ equals the sum of the weights of $2$-local gates, arranged in any architecture, that implement $U$, wherein each gate $U_j$ is weighted by its strength $W(U_j)$, defined through
\begin{equation}
W(U) :=\inf \left\{ \|h\|: U=e^{ih}\right\}.
\end{equation}
\end{definition}
\noindent
The weighted circuit complexity ${\cal C}_{\mathrm{w}}(U)$ turns out to equal the cost $C_1 (U)$ for \new{any} given unitary. We can grasp this result by Trotter-approximating the time-dependent parameterized
curve in the definition of $C_1 (U)$. 
\begin{lemma}[Weighted circuit complexity and cost] 
If $n$ denotes an integer and $U\in \mathrm{SU}(2^n)$, then
\begin{equation}
{\cal C}_{\mathrm{w}}(U) = C_1(U).
\end{equation}
\end{lemma}
Therefore, the weighted circuit complexity 
grows like the cost $C_1$.
By implication, the circuit complexity's growth will be reflected by a notion of circuit complexity that weighs the quantum
gates according to their strengths. Again, once the main text's approximate circuit complexity 
is established with an $n$-independent approximation error, one finds bounds on the weighted circuit complexity, as well. 

The last important notion of
circuit complexity that has arisen in the recent literature is that of Ref.~\cite{brandao2019models}. 
Denote by ${\cal G}_a\subset $SU$(2^{2n})$ the set of $2n$-qubit unitary circuits comprised of $\leq a$ elementary quantum gates, wherein the first $n$ qubits form the actual system and the next $n$ qubits form a memory. Let ${\cal M}_b$ denote the class of all two-outcome measurements, defined on $2n$ qubits, that require quantum circuits whose implementation requires $\leq b$ elementary quantum gates. Define
\begin{eqnarray}
    \beta(r,U)&:=& \text{maximize }\left|{\rm tr}
    \left(M  \left\{ 
    [U\otimes \id]|\phi\rangle\langle\phi| 
    [U\otimes \id]^\dagger
    - [\id/2^n \otimes 
    {\rm tr}_1 (|\phi\rangle\langle\phi|) ]
    \right\} \right)\right|,\\
    && \text{subject to } M\in {\cal M}_b,\,
    |\phi\rangle = V|0^{2n}\rangle, \, V\in {\cal G}_a, \, r=a+b.
\end{eqnarray}
In terms of this quantity, Ref.~\cite{brandao2019models} defined strong unitary complexity.

\begin{definition}[Strong unitary complexity \cite{brandao2019models}] 
Let $r\in \rr$ and $\delta\in (0,1)$.
A unitary $U \in$ SU$(2^n)$ 
has strong unitary complexity $\leq r$ if
\begin{equation}
    \beta(r,U) \geq 1-\frac{1}{2^{2n}}-\delta,
\end{equation}
denoted by 
$\tilde {\cal C}(U,\delta)\geq r$.
\end{definition}
While seemingly technically involved, the definition is operational.
The definition is also more 
stringent and demanding than more-traditional definitions of approximate circuit complexity.
To concretize this statement, we denote the diamond norm by $\|.\|_\diamond$ \cite{Watrous1}. 
\begin{lemma}[Implications of strong unitary complexity \cite{brandao2019models}] 
Suppose that $U\in U(2^n)$ 
obeys $\tilde {\cal C}(U,\delta) 
\geq r + 1$ 
for some 
$\delta\in (0,1)$,
$r\in \rr$, arbitrary
measurement procedures that include the Bell measurement. Then
\begin{equation}
\min_{\compU(V)\leq r}
\frac{1}{2}
\|{\cal U}- {\cal V}\|_\diamond>\sqrt{\delta} \, .
\end{equation}
That is, it is impossible to accurately approximate $U$ with circuits $V$ of $< r$ elementary quantum gates.
\end{lemma}
\noindent
${\cal U}$ and ${\cal V}$ denote the unitary quantum channels defined by ${\cal U}(\rho) = 
    U\rho U^\dagger$
    and
    ${\cal V}(\rho) = 
    V\rho V^\dagger$.
The diamond norm between them is
\begin{eqnarray}
\frac{1}{2}
\|{\cal U}- {\cal V}\|_\diamond &=& 
\frac{1}{2}
\sup_\rho
\|(U\otimes \id) \rho
(U\otimes \id)^\dagger-
(V\otimes \id) \rho
(V\otimes \id)^\dagger\|_1\\
&\leq &
\frac{1}{2}
\sup_\rho
\|[(U-V)\otimes \id] \rho
(U\otimes \id)^\dagger\|_1+
\frac{1}{2}
\sup_\rho
\| (U\otimes \id)\rho[
(U-V)\otimes \id]^\dagger\|_1.
\nonumber
\end{eqnarray}
We have added and subtracted a term and have used the triangle inequality.
Therefore,
\begin{eqnarray}
\frac{1}{2}
\|{\cal U}- {\cal V}\|_\diamond
\leq \frac{1}{2}
\|U-V\|_\infty 
\left(\sup_\rho
\| \rho (U\otimes \id)^\dagger\|_1
+ \sup_\rho
\| (V\otimes \id) \rho\|_1
\right)
\leq
\|U-V\|_\infty, 
\end{eqnarray}
as the operator norm is a weakly unitarily invariant norm. Therefore,
$\tilde {\cal C}(U,\delta) 
\geq r + 1$
implies that 
${\cal C}(U,\delta) \geq r$. That is, the strong unitary complexity of Ref.\ \cite{brandao2019models} is tighter than approximate circuit complexity. A topic of future work will be the exploration of the growth of approximate notions of complexity with an approximation error independent of the system size.


\begin{thebibliography}{10}

\bibitem{TComment}
The run-time of the best known algorithms for the $T$-count
  \cite{TCount}---deciding whether the optimal gate decomposition of a circuit
  presented as a sequence of Clifford gates and $T$ gates on $n$ qubits
  involves $\leq m$ $T$ gates---is $O(N^{m}\text{poly}(m,N))$, with $N:= 2^n$.

\bibitem{suppmaterial}
See the Supplementary Material at [Insert URL].

\bibitem{footnoteclifford}
Consider conjugating an arbitrary $n$-qubit Pauli operator $P$ with a uniformly
  random Clifford operator $C$. The result, $C^\dag P C$, is a uniformly random
  $n$-qubit Pauli
  operator~\cite{cleve_near-linear_2015,Webb3Design,PhysRevA.96.062336}.
  Therefore, for every initial Pauli operator $P$ and every final Pauli
  operator, some Clifford operator $C$ maps one to the other.

\bibitem{aaronson2004improved}
S.~Aaronson and D.~Gottesman.
\newblock Improved simulation of stabilizer circuits.
\newblock {\em Phys. Rev. A}, 70:052328, 2004.

\bibitem{Supremacy}
F.~Arute et~al.
\newblock Quantum supremacy using a programmable superconducting processor.
\newblock {\em Nature}, 574:505–510, 2019.

\bibitem{bochnak2013real}
J.~Bochnak, M.~Coste, and M.-F. Roy.
\newblock {\em Real algebraic geometry}, volume~36.
\newblock Springer Science \& Business Media, 2013.

\bibitem{bolt1961cliffordI}
B.~Bolt, T.~G. Room, and G.~E. Wall.
\newblock {On the Clifford collineation, transform and similarity groups. i.}
\newblock {\em J. Austr. Math. Soc.}, 2:60--79, 1961.

\bibitem{bolt1961cliffordII}
B.~Bolt, T.~G. Room, and G.~E. Wall.
\newblock {On the Clifford collineation, transform and similarity groups. II.}
\newblock {\em J. Austr. Math. Soc.}, 2:80--96, 1961.

\bibitem{bouland2019computational}
A.~Bouland, B.~Fefferman, and U.~Vazirani.
\newblock {Computational pseudorandomness, the wormhole growth paradox, and
  constraints on the AdS/CFT duality}.
\newblock {\em arXiv:1910.14646}, 2019.

\bibitem{brandao_local_2016}
F.~G. S.~L. Brand\~ao, A.~W. Harrow, and M.~Horodecki.
\newblock Local {random} {quantum} {circuits} are {approximate}
  {polynomial}-{designs}.
\newblock {\em Commun. Math. Phys.}, 346:397--434, 2016.

\bibitem{brandao2019models}
F.~G. S.~L. Brand{\~a}o, W.~Chemissany, N.~Hunter-Jones, R.~Kueng, and
  J.~Preskill.
\newblock Models of quantum complexity growth.
\newblock {\em arXiv:1912.04297}, 2019.

\bibitem{brandao_efficient_2016}
F.~G. S.~L. Brandao, A.~W. Harrow, and M.~Horodecki.
\newblock Efficient quantum pseudorandomness.
\newblock {\em Phys. Rev. Lett.}, 116, 2016.

\bibitem{brown2016complexity}
A.~R. Brown, D.~A. Roberts, L.~Susskind, B.~Swingle, and Y.~Zhao.
\newblock Complexity, action, and black holes.
\newblock {\em Phys. Rev. D}, 93:086006, 2016.

\bibitem{PhysRevLett.116.191301}
A.~R. Brown, D.~A. Roberts, L.~Susskind, B.~Swingle, and Y.~Zhao.
\newblock Holographic complexity equals bulk action?
\newblock {\em Phys. Rev. Lett.}, 116:191301, 2016.

\bibitem{PhysRevD.97.086015}
A.~R. Brown and L.~Susskind.
\newblock Second law of quantum complexity.
\newblock {\em Phys. Rev. D}, 97:086015, 2018.

\bibitem{SwingleScrambling}
B.Swingle.
\newblock Unscrambling the physics of out-of-time-order correlators.
\newblock {\em Nature Phys.}, 14:988, 2018.

\bibitem{calderbank1998quantum}
A.~R. Calderbank, E.~M. Rains, P.~M. Shor, and N.~J.~A. Sloane.
\newblock Quantum error correction via codes over gf (4).
\newblock {\em IEEE Trans. Inf. Th.}, 44:1369--1387, 1998.

\bibitem{calderbank1997quantum}
A.~R. Calderbank, E.~M. Rains, P.~W. Shor, and N.~J.~A. Sloane.
\newblock Quantum error correction and orthogonal geometry.
\newblock {\em Phys. Rev. Lett.}, 78:405, 1997.

\bibitem{Chan_19_Unitary}
A.~Chan, R.~M. Nandkishore, M.~Pretko, and G.~Smith.
\newblock Unitary-projective entanglement dynamics.
\newblock {\em Phys. Rev. B}, 99:224307, Jun 2019.

\bibitem{BigComplexity}
S.~Chapman, J.~Eisert, L.~Hackl, M.~P. Heller, R.~Jefferson, H.~Marrochio, and
  R.~C. Myers.
\newblock Complexity and entanglement for thermofield double states.
\newblock {\em SciPost Phys.}, 6:034, 2019.

\bibitem{cleve_near-linear_2015}
R.~Cleve, D.~Leung, L.~Liu, and C.~Wang.
\newblock Near-linear constructions of exact unitary 2-designs.
\newblock {\em Quant. Inf. Comp.}, 16:0721--0756, 2016.

\bibitem{dankert_exact_2009}
C.~Dankert, R.~Cleve, J.~Emerson, and E.~Livine.
\newblock Exact and approximate unitary 2-designs and their application to
  fidelity estimation.
\newblock {\em Phys. Rev. A}, 80:012304, 2009.

\bibitem{dowling2008geometry}
M.~R. Dowling and M.~A. Nielsen.
\newblock The geometry of quantum computation.
\newblock {\em Quant. Inf. Comp.}, 8:861--899, 2008.

\bibitem{Entanglement}
J.~Eisert.
\newblock Entangling power and quantum circuit complexity.
\newblock {\em arXiv:2104.03332}, 2021.

\bibitem{1408.5148}
J.~Eisert, M.~Friesdorf, and C.~Gogolin.
\newblock Quantum many-body systems out of equilibrium.
\newblock {\em Nature Phys.}, 11:124--130, 2015.

\bibitem{TCount}
D.~Gosset, V.~Kliuchnikov, M.~Mosca, and V.~Russo.
\newblock An algorithm for the $t$-count.
\newblock {\em Quant. Inf. Comp.}, 14:1277--1301, 2014.

\bibitem{gottesman1997stabilizer}
D.~Gottesman.
\newblock Stabilizer codes and quantum error correction.
\newblock {\em quant-ph/9705052}, 1997.

\bibitem{Gottesman_99_Fault}
D.~Gottesman.
\newblock Fault-tolerant quantum computation with higher-dimensional systems.
\newblock {\em Chaos, Sol. Frac.}, 10:1749--1758, 1999.

\bibitem{Gottesman_99_Heisenberg}
D.~Gottesman.
\newblock {The Heisenberg representation of quantum computers}.
\newblock In S.~P. Corney, R.~Delbourgo, and P.~D. Jarvis, editors, {\em
  Proceedings of the XXII International Colloquium on Group Theoretical Methods
  in Physics}, 1999.

\bibitem{gross_evenly_2007}
D.~Gross, K.~M.~R. Audenaert, and J.~Eisert.
\newblock Evenly distributed unitaries: on the structure of unitary designs.
\newblock {\em J. Math. Phys.}, 48:052104, 2007.

\bibitem{haferkamp2020improved}
J.~Haferkamp and N.~Hunter-Jones.
\newblock Improved spectral gaps for random quantum circuits: large local
  dimensions and all-to-all interactions.
\newblock {\em arXiv:2012.05259}, 2020.

\bibitem{haferkamp2020quantum}
J.~Haferkamp, F.~Montealegre-Mora, M.~Heinrich, J.~Eisert, D.~Gross, and
  I.~Roth.
\newblock Quantum homeopathy works: Efficient unitary designs with a
  system-size independent number of non-{Clifford gates}.
\newblock {\em arXiv:2002.09524}, 2020.

\bibitem{hall2015lie}
B.~Hall.
\newblock {\em Lie groups, Lie algebras, and representations: an elementary
  introduction}, volume 222.
\newblock Springer, 2015.

\bibitem{HaydenBlackHoles}
P.~Hayden and J.~Preskill.
\newblock Black holes as mirrors: quantum information in random subsystems.
\newblock {\em JHEP}, 0709:120, 2007.

\bibitem{hunter2019unitary}
N.~Hunter-Jones.
\newblock Unitary designs from statistical mechanics in random quantum
  circuits.
\newblock {\em arXiv:1905.12053}, 2019.

\bibitem{knill1995approximation}
E.~Knill.
\newblock Approximation by quantum circuits.
\newblock {\em quant-ph/9508006}, 1995.

\bibitem{lee2013smooth}
J.~M. Lee.
\newblock Smooth manifolds.
\newblock In {\em Introduction to smooth manifolds}, pages 1--31. Springer,
  2013.

\bibitem{Li_19_Measurement}
Y.~Li, X.~Chen, and M.~P.~A. Fisher.
\newblock Measurement-driven entanglement transition in hybrid quantum
  circuits.
\newblock {\em Phys. Rev. B}, 100:134306, 2019.

\bibitem{Shenker}
J.~Maldacena, S.~H. Shenker, and D.~Stanford.
\newblock A bound on chaos.
\newblock {\em JHEP}, 1608:106, 2016.

\bibitem{Eternal}
J.~M. Maldacena.
\newblock Eternal black holes in anti-de sitter.
\newblock {\em JHEP}, 04:021, 2003.

\bibitem{milne2017algebraic}
J.~S. Milne.
\newblock {\em Algebraic groups: the theory of group schemes of finite type
  over a field}, volume 170.
\newblock Cambridge University Press, 2017.

\bibitem{PhysRevX.8.021014}
A.~Nahum, S.~Vijay, and J.~Haah.
\newblock Operator spreading in random unitary circuits.
\newblock {\em Phys. Rev. X}, 8:021014, 2018.

\bibitem{nakata2017efficient}
Y.~Nakata, C.~Hirche, M.~Koashi, and A.~Winter.
\newblock {Efficient quantum pseudorandomness with nearly time-independent
  Hamiltonian dynamics}.
\newblock {\em Phys. Rev. X}, 7:021006, 2017.

\bibitem{neill_blueprint_2017}
C.~Neill, P.~Roushan, K.~Kechedzhi, S.~Boixo, S.~V. Isakov, V.~Smelyanskiy,
  R.~Barends, B.~Burkett, Y.~Chen, and Z.~Chen.
\newblock A blueprint for demonstrating quantum supremacy with superconducting
  qubits.
\newblock 2017.

\bibitem{nielsen2005geometric}
M.~A. Nielsen.
\newblock A geometric approach to quantum circuit lower bounds.
\newblock {\em quant-ph/0502070}, 2005.

\bibitem{NielsenChuang}
M.~A. Nielsen and I.~L. Chuang.
\newblock {\em Quantum computation and quantum information}.
\newblock Cambridge Series on Information and the Natural Sciences. Cambridge
  University Press, 2000.

\bibitem{Nielsen_06_Optimal}
M.~A. Nielsen, M.~R. Dowling, M.~Gu, and A.~C. Doherty.
\newblock Optimal control, geometry, and quantum computing.
\newblock {\em Phys. Rev. A}, 73:062323, 2006.

\bibitem{nielsen2006quantum}
M.~A. Nielsen, M.~R. Dowling, M.~Gu, and A.~C. Doherty.
\newblock Quantum computation as geometry.
\newblock {\em Science}, 311:1133--1135, 2006.

\bibitem{Note1}
A $t$-design is a probability distribution, over unitaries, whose first $t$
  moments equal the Haar measure's moments~\cite
  {gross_evenly_2007,dankert_exact_2009,brandao_local_2016}. {\protect \color
  {dg}The Haar measure is the unique unitarily invariant probability measure
  over a compact group.}

\bibitem{Note2}
Technically, this bound on the dimension does not follow from Lemma~2 as it
  does for the bound $\protect \qopname \relax o{dim}F^{A'}(\protect \mathrm
  {SU}(4)^{\times R'})\leq 15 R'$. The reason is, the quotient space $\protect
  \mathrm {SU}(4)^{\times R'}/SU(2)^{\times (2R'-n)}$ is not necessarily
  semialgebraic. This difficulty can be resolved via Sard's theorem~\cite
  {sard1965hausdorff}, which asserts, as a special case, that the Hausdorff
  dimension of a smooth map's image is bounded by its domain's dimension. A
  semialgebraic set's dimension is the greatest dimension in its stratification
  and so agrees with the Hausdorff dimension.

\bibitem{Note3}
Technically, the \protect \emph {sets of few gates} form a set of measure 0 in
  SU$(4)^{\times R}$. Circuits are not in SU$(4)^{\times R}$, as explained in
  Fig.~\ref {fig_Contraction_map}. However, we expected the sentence above to
  be more intuitive with ``short circuits'' instead of ``sets of few gates.''.

\bibitem{Note4}
Again, by ``short circuits,'' we mean, ``sets of few gates.'' We replaced the
  latter phrase for ease of expression.

\bibitem{RandomHamiltonians}
E.~Onorati, O.~Buerschaper, M.~Kliesch, W.~Brown, A.~H. Werner, and J.~Eisert.
\newblock {Mixing properties of stochastic quantum Hamiltonians}.
\newblock {\em Commun. Math. Phys.}, 355:905, 2017.

\bibitem{ngupta_Silva_Vengalattore_2011}
A.~Polkovnikov, K.~Sengupta, A.~Silva, and M.~Vengalattore.
\newblock Nonequilibrium dynamics of closed interacting quantum systems.
\newblock {\em Rev.\ Mod.\ Phys.}, 83:863--883, 2011.

\bibitem{PhysRevLett.106.170501}
D.~Poulin, A.~Qarry, R.~Somma, and F.~Verstraete.
\newblock {Quantum simulation of time-dependent Hamiltonians and the convenient
  illusion of Hilbert space}.
\newblock {\em Phys. Rev. Lett.}, 106:170501, 2011.

\bibitem{roberts2017chaos}
D.~A. Roberts and B.~Yoshida.
\newblock Chaos and complexity by design.
\newblock {\em JHEP}, 2017:121, 2017.

\bibitem{sard1965hausdorff}
A.~Sard.
\newblock Hausdorff measure of critical images on {Banach} manifolds.
\newblock {\em Am. J.f Math.}, 87:158--174, 1965.

\bibitem{shenker2014black}
S.~H. Shenker and D.~Stanford.
\newblock Black holes and the butterfly effect.
\newblock {\em JHEP}, 2014(3):1--25, 2014.

\bibitem{shenker2014multiple}
S.~H. Shenker and D.~Stanford.
\newblock Multiple shocks.
\newblock {\em JHEP}, 2014(12):1--20, 2014.

\bibitem{Skinner_19_Measurement}
B.~Skinner, J.~Ruhman, and A.~Nahum.
\newblock Measurement-induced phase transitions in the dynamics of
  entanglement.
\newblock {\em Phys. Rev. X}, 9:031009, Jul 2019.

\bibitem{stanford2014complexity}
D.~Stanford and L.~Susskind.
\newblock Complexity and shock wave geometries.
\newblock {\em Phys. Rev. D}, 90:126007, 2014.

\bibitem{EntanglementNotEnough}
L.~Susskind.
\newblock Entanglement is not enough.
\newblock 2014.
\newblock arXiv:1411.0690.

\bibitem{susskind2016computational}
L.~Susskind.
\newblock Computational complexity and black hole horizons.
\newblock {\em Fort. Phys.}, 64:24--43, 2016.

\bibitem{susskind2018black}
L.~Susskind.
\newblock Black holes and complexity classes.
\newblock {\em arXiv:1802.02175}, 2018.

\bibitem{Watrous1}
J.~Watrous.
\newblock Semidefinite programs for completely bounded norms.
\newblock {\em Th. Comp.}, 5, 2009.

\bibitem{Webb3Design}
Z.~Webb.
\newblock {The Clifford group forms a unitary 3-design}.
\newblock 2015.
\newblock arXiv:1510.02769.

\bibitem{resourcepaper}
N.~{Yunger Halpern}, N.~B.~T. {Kothakonda}, J.~{Haferkamp}, A.~{Munson},
  J.~{Eisert}, and P.~{Faist}.
\newblock {Resource theory of quantum uncomplexity}.
\newblock {\em arXiv e-prints}, page arXiv:2110.11371, 2021.

\bibitem{PhysRevA.96.062336}
H.~Zhu.
\newblock {Multi-qubit Clifford groups are unitary 3-designs}.
\newblock {\em Phys. Rev. A}, 96:062336, 2017.

\end{thebibliography}
\end{document}